\documentclass[12pt]{article}
\usepackage{latexsym,amsmath,amsfonts,natbib,bm}
\usepackage{amssymb,amsthm,rotating,threeparttable}
\usepackage{algorithm,algorithmicx,algpseudocode,epsfig}
\usepackage{subfigure}
\usepackage{graphicx}
\usepackage{amsmath}
\usepackage{setspace}
\usepackage{lscape}
\usepackage{atbegshi}
\usepackage{url}
\usepackage{undertilde}

\usepackage[pagewise,mathlines]{lineno}
\usepackage[textwidth=6in,textheight=8.5in,centering]{geometry}
\usepackage{xr}
\newcommand{\blind}{0}

\onehalfspace

\include{def}
\def\dispmuskip{\thinmuskip= 3mu plus 0mu minus 2mu \medmuskip=  4mu plus 2mu minus 2mu \thickmuskip=5mu plus 5mu minus 2mu}
\def\textmuskip{\thinmuskip= 0mu                    \medmuskip=  1mu plus 1mu minus 1mu \thickmuskip=2mu plus 3mu minus 1mu}
\def\beq{\dispmuskip\begin{equation}}    \def\eeq{\end{equation}\textmuskip}
\def\beqn{\dispmuskip\begin{displaymath}}\def\eeqn{\end{displaymath}\textmuskip}
\def\bea{\dispmuskip\begin{eqnarray}}    \def\eea{\end{eqnarray}\textmuskip}
\def\bean{\dispmuskip\begin{eqnarray*}}  \def\eean{\end{eqnarray*}\textmuskip}

\newtheorem{theorem}{Theorem}

\newtheorem{lemma}{Lemma}

\newtheorem{proposition}{Proposition}
\newtheorem{remark}{Remark}

\newtheorem{ISalgorithm}{Algorithm}
\newtheorem{assumption}{Assumption}

\newcommand{\eps}{\epsilon}
\newcommand{\wh}{\widehat}
\newcommand{\wt}{\widetilde}

\def\E{\mathbb{E}}                         

\def\a{\alpha}

\def\s{\sigma}

\def\t{\theta}

\def\N{{\cal N}}

\def\ESS{\text{\rm ESS}}

\def\CT{\text{\rm CT}}

\def\Var{\mathbb{V}}

\def\IS{\text{\rm IS}}

\def\Sup{\text{\rm Supp}}

\newcommand{\bq}[1]{\begin{equation}\label{#1}}
\newcommand{\eq}{\end{equation}}
\newcommand{\bqn}{\begin{eqnarray}}
\newcommand{\eqn}{\end{eqnarray}}
\newcommand{\bqns}{\begin{eqnarray*}}
\newcommand{\eqns}{\end{eqnarray*}}

\newcommand{\dd}{\,\textrm{d}}
\newcommand{\Ee}{\textrm{\rm E}}
\newcommand{\Nn}{\textrm{\rm N}}
\newcommand{\pitilde}{\widetilde \pi}
\newcommand{\ov}{\overline}

\def\E{\mathbb{E}}
\def\eps{\varepsilon}
\def\({\left(}
\def\){\right)}

\makeatletter
\def\@seccntformat#1{\@ifundefined{#1@cntformat}%
   {\csname the#1\endcsname\quad}  
   {\csname #1@cntformat\endcsname}
}
\let\oldappendix\appendix 
\renewcommand\appendix{%
    \oldappendix
    \newcommand{\section@cntformat}{\appendixname~\thesection\quad}
}
\makeatother
\usepackage{cleveref} 

\begin{document}
\doublespacing
\if0\blind
{
\title{Importance Sampling Squared for Bayesian Inference and Model Choice with Estimated Likelihoods}
\author{
M.-N. Tran\thanks{The University of Sydney Business School}
\and M. Scharth\footnotemark[1]
\and M. K. Pitt\thanks{Department of Mathematics, King's College London}
\and R. Kohn\thanks{UNSW Business School, University of New South Wales}
\thanks{The research of Minh-Ngoc Tran, Marcel Scharth and Robert Kohn
was partially supported by Australian Research Council grant DP12120104014. The authors
thank Denzil Fiebig for sharing the pap smear data.}
}
\maketitle
} \fi

\if1\blind
{
  \begin{center}
   \title {Importance Sampling Squared for Bayesian Inference and Model Choice with Estimated Likelihoods}
\end{center}
  \medskip
  \maketitle
} \fi

\vspace{-8mm}
\begin{abstract}
\noindent
We propose an approach to Bayesian inference that uses importance sampling to generate the parameters for models where the
likelihood is analytically intractable but can be estimated unbiasedly. We refer to this procedure as importance sampling squared (IS$^2$),
as we can often estimate the likelihood by importance sampling or sequential importance sampling. The IS$^2$ method leads to efficient estimates of
expectations with respect to the posterior and their Monte Carlo standard errors when a suitable proposal for the parameters is available.
A key motivation for the IS$^2$ method is that we can use it as a tool for estimating the marginal likelihood (and the standard error of the estimator), irrespective of whether we use IS$^2$ or Markov chain Monte Carlo (MCMC) to estimate the model. The marginal likelihood is a  fundamental tool in Bayesian model choice, but estimating it has proved difficult or computationally expensive using other approaches. Our article formally justifies the IS$^2$ method and studies its convergence properties. We analyze the effect of estimating the likelihood on the resulting inference and provide guidelines on how to determine
the precision of the likelihood estimator in order to obtain an optimal tradeoff between computational
cost and accuracy for posterior inference on the model parameters. The advantages of the IS$^2$ procedure are illustrated empirically for a generalized multinomial logit model and a stochastic volatility model.

\end{abstract}
\doublespacing

\section{Introduction} \label{Sec:introduction}
A wide range of statistical models, such as nonlinear and non-Gaussian state space models and generalized linear mixed models, lead to analytically intractable likelihood functions. When the density of the observations conditional on the parameters and a vector of latent
variables is available in closed form, we can use importance sampling (IS) and, more generally, sequential importance sampling
to estimate the likelihood unbiasedly.  Our article considers importance sampling for Bayesian inference when working with an estimated likelihood. We call this procedure importance sampling squared (IS$^2$).

A key motivation for the IS$^2$ method is that it estimates the marginal likelihood (and the standard error of this estimator) accurately and efficiently in cases where only an estimate of the likelihood is available.
The marginal likelihood is a fundamental tool for Bayesian model comparison \citep{Kass:1995}, but it is challenging to use the output of Markov Chain Monte Carlo (MCMC) methods to estimate it, especially for models with an intractable likelihoods \citep{Chib:2001,Perrakis:2014}. The IS$^2$ method obtains the marginal likelihood automatically when estimating expectations with respect to the posterior. Moreover, the method can be used as a tool for marginal likelihood estimation even if we estimate the posterior distribution by MCMC methods. In this case, it is typically straightforward to use the MCMC output to form a proposal density for the parameters to use in the IS$^2$ method to  estimate the marginal likelihood.

Our article shows that IS is still valid for estimating expectations with respect to the posterior
when the likelihood is estimated unbiasedly,
and prove a law of large numbers and a central limit theorem for these estimators. This analysis relates directly
 to the results in \cite{Fearnhead2008} and \cite{Fearnhead2010}, who considered random weight importance sampling in the context of particle filtering. Our results allow us to analyze how much asymptotic efficiency is lost
when working with an estimated likelihood by comparing the asymptotic variance obtained under IS$^2$ to the standard case where the likelihood is known.
We show that the ratio of the asymptotic variance of the IS$^2$ estimator to the asymptotic variance of the IS estimator,
which we call the inflation factor, is greater than or equal to 1,
and is equal to 1 if and only if the  likelihood is known.
The inflation factor increases exponentially with the variance of the estimator of the log-likelihood.

A critical implementation issue is the choice of the number of particles $N$ for estimating the likelihood.
A large $N$ gives a more accurate estimate of the likelihood
at greater computational cost, while a small $N$ can lead to an estimator with  a very large variance.
We provide theoretical and practical guidelines on how to select $N$ to obtain an optimal tradeoff between
accuracy and computational cost. Our results show that the efficiency of $\IS^2$
is weakly sensitive
to the number of particles around the optimal value. Moreover,
the loss of efficiency decreases at
worst linearly when we choose $N$ higher than the optimal value, whereas the efficiency can deteriorate
exponentially when $N$ deviates appreciably below its optimal value.  We therefore advocate a conservative choice of $N$ in practice.
We propose two approaches for selecting the number of particles. The first approach is static because
it selects the same number of particles for all
parameter values. The second approach is dynamic because it selects the optimal number of particles depending on the
 parameter values. We show both theoretically, in section \ref{SSS: static dynamic number of particles} of the supplementary material,
 and empirically, in section~\ref{sec:likeval1} of the supplementary material,
that the dynamic approach can be much more efficient than the static approach.


Our method relates to alternative approaches to Bayesian inference for models with intractable likelihoods. \cite{Beaumont:2003} develops a pseudo marginal Metropolis Hastings (PMMH) scheme
to carry out Bayesian inference with an estimated likelihood.
\cite{Andrieu:2009} formally study Beaumont's method and give conditions under which the chain converges. \cite{Andrieu:2010} use MCMC for inference in state
space models where the likelihood is estimated by the particle filter,
and \cite{Pitt:2012} and \cite{Doucet:2014} discuss the issue of the optimal number of particles
to be used in likelihood estimation.
The  SMC$^2$ method of \cite{Chopin:2012} is a sequential Monte Carlo procedure
for inference in space state models in which intractable likelihood contributions are estimated by random weight particle filters. The initialisation step in their algorithm corresponds to IS$^2$ for an initial set of observations using the prior as a proposal for the parameters. 

Given a statistical model with an intractable likelihood, we can choose to carry out off-line Bayesian inference using either the IS$^2$ or MCMC as in \cite{Andrieu:2010} and \cite{Pitt:2012}. There can be several advantages in following the IS$^2$ approach. First, IS$^2$ estimates the marginal likelihood.
Second, it is straightforward to obtain the MC standard error of IS$^2$ estimators since they are based on independent draws. In contrast, it can be difficult to obtain precise standard errors for estimators based on MCMC output due to autocorrelation in the Markov Chain. Third, it is straightforward to fully parallelize the IS$^2$ procedure so that it can be more computationally attractive to use IS$^2$ than PMMH when it is expensive to estimate the likelihood. Fourth, it is simple to implement variance reduction methods such as antithetic sampling, stratified mixture sampling and control variates for IS$^2$, as well as to use randomized Quasi-Monte Carlo techniques to improve numerical efficiency; see the applications in section~\ref{Sec:Examples}. Fifth, MCMC requires computationally expensive burn-in draws and assessing whether the Markov chain has converged and
loses the information from rejected values (in standard implementations). IS$^2$  uses all the draws from the proposal for the parameters. Sixth, high variance likelihood estimates can lead the PMMH Markov Chain to get trapped in certain regions of the extended sampling space, while isolated importance weights in IS$^2$ can be directly diagnosed and stabilised using methods such as Pareto smoothed importance sampling \citep{Vehtari2015}. Finally, IS$^2$ can deal with a multimodal posterior while PMMH methods can often get trapped in local modes.

Although the SMC$^2$ method is designed for sequential updating, it can also be used effectively for off-line inference to deal with multimodal posteriors, for example.
However, for off-line inference IS$^2$ has advantages over SMC$^2$.  First, the sequential nature of the SMC$^2$ algorithm can lead to a large implementation and computational effort compared to IS$^2$. Second, IS$^2$ provides estimates of the Monte Carlo standard errors for the posterior estimates, while it is not easy to do so with SMC$^2$. Third, the IS$^2$ method can take advantage of efficient off-line methods to estimate the likelihood, e.g. in state space models, which can outperform online likelihood estimation, as in SMC$^2$, by orders of magnitude \citep[e.g.][]{PEIS}. Third, because it is sequential, SMC$^2$ can perform poorly in time series models with structural breaks or other interventions, unlike the PMMH and IS$^2$ methods. However, we view IS$^2$ and SMC$^2$ as being {\em complementary} when used for prediction: we can use SMC$^2$ for real time updating after starting from off-line IS$^2$ estimates.

We illustrate the IS$^2$ method in empirical applications for the generalized multinomial logit (GMNL)
model of \cite{fklw2010} and a two factor stochastic volatility (SV) model with leverage effects.
Our results for the GMNL model show that the IS$^2$ approach is numerically more efficient than PMMH.
In the standard implementation, the IS$^2$ method leads 65-89\% lower Monte Carlo mean-squared errors (MSEs) for estimating the posterior means compared to PMMH. When incorporating Quasi-Monte Carlo techniques, we obtain 87\% to 99\% reductions in MSE over the PMMH method. We show that the method for optimally selecting the number of particles in estimating the GMNL model leads to
improved performance. The IS$^2$ method accurately estimates the posterior distribution under our optimal implementation.

The SV application is based on daily returns of the S\&P 500 index between 1990 and 2012. We show that the variance of the log-likelihood estimates based on the particle efficient importance sampling method of \cite{PEIS} is small for this problem, despite the long time series. Hence, IS$^2$ leads to highly accurate estimates of the posterior expectations for this example in a short amount of computing time. As suggested by the theory, the efficiency of the IS$^2$  method is insensitive to the choice of $N$ in this example. As few as two particles for estimating the likelihood (including an antithetic draw) lead to an efficient procedure for this model, highlighting the practical convenience of the approach.

There is an online supplement to the article which contains results that complement those in the main paper.
References to the main paper are of the form section~2, equation~(2) and assumption~2, etc,
whereas for the supplementary material we use section~S2, equation (S2) and assumption S2, etc.

\section{Latent variable models}\label{Sec:latent variable models}
This section sets out the main class of models with latent variables that we wish to estimate.
However, the application of the method presented in this article is not limited to this class of models and section~\ref{sec:sv} applies the
IS$^2$ method to a time series stochastic volatility model. The method can also be applied in a variety of other contexts, including
panel data estimation for large datasets using the subsampling approach of \cite{quiroz:2015}.

Suppose there are $n$ individuals. Individual $i$ has $T_i$ observations
$y_{i}=\{y_{i1},...,y_{iT_i}\}$,  and for each  $y_{it}$, there is a vector of covariates
 $x_{it}$. Let $y = \{y_{1}, \dots, y_n \}$. We associate a latent vector $\alpha_{i}$ with individual $i$ and denote the
 vector of unknown parameters in the model by $\theta$. Assuming that the individuals are independent, the joint density of the latent
variables and the observations is
\beqn
p(y,\alpha|\theta) =\prod_{i=1}^{n}p_i(y_{i},\alpha_{i}|\theta)=\prod_{i=1}^{n}p_i(y_{i}|\alpha_{i},\theta)p(\alpha_{i}|\theta).
\eeqn
By assigning a prior $p(\theta)$ for $\theta$, standard MCMC approaches can be used to sample from the joint posterior
$p(\alpha,\theta|y)\propto p(y,\alpha|\theta)p(\theta)$
by cycling between $p(\alpha|\theta, y)$ and $p(\theta|\alpha, y )$. Typically, such models are estimated by introducing auxiliary latent variables
as in \cite{albert:chib:1993, polson:scott:windle:2013} for example.
These methods require computationally expensive burn-in draws and the assessment of the convergence of the Markov chain.
More importantly, they also suffer from the problem of slow mixing and are inefficient when the latent vector $\alpha$ is high dimensional or
when the data is unbalanced  \citep[see][]{johndrow2016inefficiency}.

In this article, we are interested in inference when the likelihood is analytically intractable
and given by
\beqn
p(y|\theta)   =\prod_{i=1}^{n}p_i(y_{i}|\theta),\quad p_i(y_{i}|\theta)   =\int p_i(y_{i}|\alpha_{i},\theta)p(\alpha_{i}|\theta)d\alpha_{i}.
\eeqn
Section \ref{Sec:Examples} considers more specific structures. We will use an
unbiased IS estimator $\widehat{p}_{N}(y|\theta)$ of $p(y|\theta)$ based upon a simulation sample size of
$N$, which we shall call the number of particles.
Let $h_i(\alpha_i|y,\t)$ be an importance density for $\a_i$.
The density $p_i(y_i|\t)$ is estimated unbiasedly by
\beq\label{e:IS_estimator}
\wh p_{N,i}(y_i|\t)=\frac{1}{N}\sum_{j=1}^Nw_i(\a_i^{(j)},\t),\;\;w_i(\a_i^{(j)},\t)=\frac{p(y_i|\a_i^{(j)},\t)p(\a_i^{(j)}|\t)}{h_i(\a_i^{(j)}|y,\t)},\;\;\a_i^{(j)}\stackrel{iid}{\sim}h_i(\cdot|y,\t).
\eeq
Hence, $\widehat{p}_{N}(y|\theta) =\prod_{i=1}^{n}\widehat{p}_{N,i}(y_{i}|\theta)$
is an unbiased estimator of the likelihood $p(y|\t)$.
We note that it is possible to use a different number of particles $N_i$ for each individual $i$ and we shall do so in the empirical examples in section~\ref{Sec:Examples}. We also consider the particle filter applied to state
space models in section~\ref{sec:sv}, which again yields an unbiased estimator of the
likelihood $\widehat{p}_{N}(y|\theta)$.

\section{The $\IS^2$ method}\label{Sec:IS}
Let $p(\t)$ be the prior for $\theta$, $p(y|\t)$ the likelihood
and $\pi(\t)\propto p(\t)p(y|\t)$ the posterior distribution defined on the space $\Theta\subset\mathbb{R}^d$.
Bayesian inference typically requires computing an integral of the form
\beq\label{e:integral}
\E_\pi(\varphi)=\int_\Theta\varphi(\t)\pi(\t)d\t,
\eeq
for some function $\varphi$ on $\Theta$ such that the integral \eqref{e:integral} exists.
When the likelihood can be evaluated, IS
is a popular method for estimating this integral \citep[see, e.g.,][]{Geweke:1989}.
Let $g_\IS(\t)$ be such an importance density.
Then,  the IS estimator of $\E_\pi(\varphi)$ is
\beq\label{e:standardISestimator}
\wh\varphi_\IS= \frac{\frac1{M}\sum_{i=1}^{M}\varphi(\t_i)w(\t_i)}{\frac1{M}\sum_{i=1}^{M}w(\t_i)},\;\;\text{with weights}\;\;w(\t_i)=\frac{p(\t_i)p(y|\t_i)}{g_\IS(\t_i)},\;\;\t_i\stackrel{iid}{\sim}g_\IS(\t).
\eeq
By \cite{Geweke:1989} and Remark~\ref{remark: IS scheme} below, $\wh\varphi_\IS\stackrel{a.s.}{\longrightarrow}\E_\pi(\varphi)$
as $M\to\infty$,
and a central limit theorem also holds for $\wh\varphi_\IS$.

When the likelihood cannot be evaluated,
standard IS cannot be used because the weights $w(\t_i)$ in \eqref{e:standardISestimator} are unavailable.
Algorithm~\ref{alg: IS2 alg} presents the $\IS^2$ scheme for estimating the integral \eqref{e:integral}
when the likelihood $p(y|\t)$ is estimated unbiasedly by $\wh p_N(y|\t)$, i.e. assuming that

\begin{assumption}\label{ass: exp}  
 $\E[\wh p_N(y|\t)]=p(y|\t)$ for every $\theta\in\Theta$, where the expectation is with respect to the random variables generated
in the process of estimating the likelihood.
\end{assumption}

\begin{ISalgorithm} [IS$^2$ algorithm] \label{alg: IS2 alg}
For $i=1,...,M, $
\begin{enumerate}
\item Generate $\t_i\stackrel{iid}{\sim}g_\IS(\t)$ and compute the likelihood estimate $\wh p_N(y|\t_i)$.
\item Compute the weight $\wt w(\t_i)={p(\t_i)\wh p_N(y|\t_i)}/{g_\IS(\t_i)}$.
\end{enumerate}
\end{ISalgorithm}
The $\IS^2$ estimator of $\E_\pi(\varphi)$ is defined as
\beq\label{e:ISestimator}
\wh\varphi_{\IS^2}:=\frac{\frac1{M}\sum_{i=1}^{M}\varphi(\t_i)\wt w(\t_i)}{\frac1{M}\sum_{i=1}^{M}\wt w(\t_i)}.
\eeq
Algorithm~\ref{alg: IS2 alg} is identical
to the standard IS scheme,
except that the likelihood is replaced by its estimate.
To see that $\wh\varphi_{\IS^2}$ is a valid estimator of $\E_\pi(\varphi)$,
we follow \cite{Pitt:2012} and write $\wh p_N(y|\t)$ as $p(y|\t)e^z$,
where $z:=\log\;\wh p_N(y|\t)-\log\;p(y|\t)$ is a random variable
whose distribution is governed by the randomness occurring when estimating the likelihood $p(y|\t)$.
Let $g_N(z|\t)$ be the density of $z$. Assumption~\ref{ass: exp} implies that
$\E(e^z|\theta)=\int_\mathbb{R} e^zg_N(z|\t)dz=1$.
We define
\begin{align}\label{e:eobietgi1}
{\pitilde}_N(\t,z) & :=p(\t)g_N(z|\t)p(y|\t)e^z/p(y) =\pi(\theta)\pi_{N}(z|\theta), \notag
\end{align}
where $\pi_{N}(z|\theta)=e^{z}g_{N}(z|\theta)$,
as the joint posterior density of $\t$ and $z$ on the extended space $\widetilde\Theta=\Theta\otimes\Bbb{R}$.
Then,  ${\pitilde}_N(\theta) = \pi(\theta)$ because
$\underline{}\int_\mathbb{R}{\pitilde}_N(\t,z)dz = p(\theta)p(y|\theta) /p(y) = \pi(\theta)$.

The integral \eqref{e:integral} can be written as
\bea\label{e:integral1}
\E_\pi(\varphi)=\int_{\wt\Theta}\varphi(\t)\pi_N(\t,z)d\t dz=\frac{1}{p(y)}\int_{\wt\Theta}\varphi(\t)\frac{p(\t)p(y|\t)e^zg_N(z|\t)}{\wt g_\IS(\t,z)}\wt g_\IS(\t,z)d\t dz,
\eea
with $\wt g_\IS(\t,z)=g_\IS(\t)g_N(z|\t)$ an importance density on $\wt\Theta$.
Let $(\t_i,z_i)\stackrel{iid}{\sim}\wt g_\IS(\t,z)$, i.e. generate $\t_i\sim g_\IS(\t)$ and then $z_i\sim g_N(z|\t_i)$.
It is straightforward to see that the estimator $\wh\varphi_{\IS^2}$ defined in \eqref{e:ISestimator}
is exactly an IS estimator of the integral defined in \eqref{e:integral1},
with importance density $\wt g_\IS(\t,z)$ and weights
\beq\label{eq:weight_IS2}
\wt w(\t_i,z_i)=\frac{p(\t_i)p(y|\t_i)e^{z_i}g_N(z_i|\t_i)}{\wt g_\IS(\t_i,z_i)}=\frac{p(\t_i)\wh p_N(y|\t_i)}{g_\IS(\t_i)}=\wt w(\t_i).
\eeq
This formally justifies Algorithm 1.
Theorem \ref{the:IS1} gives some asymptotic properties (in $M$) of $\IS^2$ estimators. Its proof is in the Appendix.

\begin{theorem}\label{the:IS1}   
Suppose that Assumption 1 holds, $\E_\pi(\varphi)$ exists and is finite,
and $\Sup(\pi)\subseteq\Sup(g_\text{IS})$, where $\Sup(\pi)$ denotes the support of the distribution $\pi$.
\begin{itemize}
\item[(i)] For any $N\geq 1$,  $\wh\varphi_{\IS^2}\stackrel{a.s.}{\longrightarrow}\E_\pi(\varphi)$ as $M\to\infty$.
\item[(ii)] If
\begin{align}\label{eq: finite var}
\int h(\theta)^2 \left (  \frac{\pi(\theta) } { g_{\IS}(\theta) } \right ) ^2  \Big ( \int  \exp(2z)g_N(z|\theta)dz \Big )  g_{\IS}(\theta)  d\theta
\end{align}
is finite for $h(\theta) =  \varphi(\theta) $ and $h(\theta) = 1 $ for all $N$, then
\beqn
\sqrt{M}\Big(\wh\varphi_{\IS^2}-\E_\pi(\varphi)\Big)\stackrel{d}{\to}\N\big(0,\sigma^2_{\IS^2}(\varphi)\big),\;\;M\to\infty,
\eeqn
where the asymptotic variance in $M$ for fixed $N$ is given by
\beq\label{e:ISvar_est_llh}
\sigma^2_{\IS^2}(\varphi)=\E_\pi\left\{\big(\varphi(\t)-\E_\pi(\varphi)\big)^2\frac{\pi(\t)}{g_\IS(\t)}\E_{g_N}[\exp(2z)]\right\}.
\eeq
\item[(iii)] Define
\beq\label{e:ISvar_est1}
\wh{\sigma^2_{\IS^2}(\varphi)}:=\frac{M\sum_{i=1}^{M}\big(\varphi(\t_i)-\wh\varphi_{\IS^2}\big)^2\wt w(\t_i)^2}{\left(\sum_{i=1}^{M}\wt w(\t_i)\right)^2}.
\eeq
If  the conditions in (ii) hold, then $\wh{\sigma^2_{\IS^2}(\varphi)}\stackrel{a.s.}{\longrightarrow}\sigma^2_{\IS^2}(\varphi)$ as $M\to\infty$, for given $N$.
\end{itemize}
\end{theorem}

We note that both $\sigma^2_{\IS^2}(\varphi)$ and $\wh{\sigma^2_{\IS^2}(\varphi)}$ depend on $N$,
but, for simplicity, we do not show this  dependence explicitly in the notation.
Here, all the probabilistic statements, such as the almost sure convergence,
must be understood on the extended probability space that takes into account the extra randomness occurring when estimating the likelihood.
The result (iii) is practically useful,
because \eqref{e:ISvar_est1} allows us to estimate the standard error of the IS$^2$ estimator.

\begin{remark} \label{remark: IS scheme} The standard IS scheme using the likelihood is  
a special case of Algorithm 1, where
$\wh p_N(y|\t)=p(y|\t)$,
the weights $\wt w(\t_i)=w(\t_i)$ and therefore $\wh\varphi_{\IS^2}=\wh\varphi_\IS$.
In this case  $g_N(z|\t)=\delta_0(z)$, the delta Dirac distribution concentrated at zero.
From Theorem \ref{the:IS1},  $\wh\varphi_\IS\stackrel{a.s.}{\longrightarrow}\E_\pi(\varphi)$,
and
\beq\label{e:ISvar_exact}
\sqrt{M}\Big(\wh\varphi_\IS-\E_\pi(\varphi)\Big)\stackrel{d}{\to}\N(0,\sigma^2_\IS(\varphi)),
\eeq
with
\beq\label{e:ISasym_var}
\sigma^2_\IS(\varphi)=\E_\pi\left\{\big(\varphi(\t)-\E_\pi(\varphi)\big)^2\frac{\pi(\t)}{g_\IS(\t)}\right\},
\eeq
which is  estimated by
\beq\label{e:ISvar_est}
\wh{\sigma^2_\IS(\varphi)}=\frac{M\sum_{i=1}^{M}\big(\varphi(\t_i)-\wh\varphi_\IS\big)^2w(\t_i)^2}{\left(\sum_{i=1}^{M}w(\t_i)\right)^2}.
\eeq
These convergence results are well known in the literature; see, e.g. \cite{Geweke:1989}.
\end{remark}

\subsection{$\IS^2$ for estimating the marginal likelihood}\label{sec:marg llh}
The marginal likelihood
$p(y) = \int_\Theta p(\t)p(y|\t)d\t$
is a fundamental tool for Bayesian model choice \citep[see, e.g.,][]{Kass:1995}.
Estimating  the marginal likelihood and its associated standard error accurately has proved difficult with standard MCMC techniques \citep{Chib:2001,Perrakis:2014}.
This problem is even more severe when the likelihood is intractable.
This section and section \ref{Subsec:tradeoffIS} discuss how to estimate $p(y)$ and the associated standard error
optimally, reliably and unbiasedly by $\IS^2$ when the likelihood is intractable.

The $\IS^2$ estimator of $p(y)$ is
\beq\label{eq:llh IS2}
\wh p_{\IS^2}(y) = \frac{1}{M}\sum_{i=1}^M \wt w(\t_i),
\eeq
with the samples $\t_i$ and the weights obtained from Algorithm 1.
The estimator of the variance of $\wh p_{\IS^2}(y)$ is
\beq\label{eq:var llh IS2}
\wh{\Var}(\wh p_{\IS^2}(y)) = \frac{1}{M}\sum_{i=1}^M \big(\wt w(\t_i)-\wh p_{\IS^2}(y)\big)^2.
\eeq
The following theorem shows some properties of these $\IS^2$ estimators.
Its proof is in the Appendix.
\begin{theorem}\label{lem: robust} Let $M$ be the number of particles for $\theta$ and $N$ be the number of particles for estimating the likelihood. Under the assumptions in Theorem \ref{the:IS1}
\begin{itemize}
  \item[(i)] $\E[\wh p_{\IS^2}(y)]=p(y)$ and $\wh p_{\IS^2}(y)\stackrel{a.s.}{\longrightarrow}p(y)$ as $M\to\infty$ for any $N\geq 1$.
  \item[(ii)] There exists a finite constant $K$ such that $\Var(\wh p_{\IS^2}(y))\leq K/M$.
  \item[(iii)] $\sqrt{M}(\wh p_{\IS^2}(y)-p(y))\stackrel{d}{\to}\N\big(0,\Var(\wh p_{\IS^2}(y))$
and $\wh{\Var}(\wh p_{\IS^2}(y))\stackrel{a.s.}{\longrightarrow}\Var(\wh p_{\IS^2}(y))$ as $M\to\infty$ for any $N\geq 1$.
\end{itemize}
\end{theorem}

\subsection{The effect on importance sampling of estimating the likelihood}\label{Subsec:influenceIS}
The results in the previous section show that
it is straightforward to use importance sampling
even when the likelihood is intractable but unbiasedly estimated.
This section addresses the  question of how much asymptotic efficiency is lost
when working with an estimated likelihood.
We follow \cite{Pitt:2012} and \cite{Doucet:2014} and
make the following idealized assumption to make it possible to develop some theory.
Proposition~\ref{prop:large sample z} of section~\ref{SS app: large sample properties}  of the supplementary material justifies this assumption for panel data.
\begin{assumption} \label{ass: normal}  
\begin{enumerate}
\item
There exists a function $\gamma^2(\t)$ such that the density $g_N(z|\t)$ of $z$ is $\N(-\frac{\gamma^2(\t)}{2N},\frac{\gamma^2(\t)}{N})$, where $\N(a,b^2)$ is a univariate normal
density with mean $a$ and variance $b^2$.
\item
 For a given $\sigma^2>0$, define $N_{\s^2}(\t):=\gamma^2(\t)/\s^2$. Then, $\Var(z|\theta, N =N_{\s^2}(\t) )\equiv\s^2$ for all $\theta \in \Theta$.
\end{enumerate}
\end{assumption}
If $g_N(z|\theta)$ is Gaussian, then, by Assumption \ref{ass: exp}, its mean must be $-\frac12 $ times its variance in order that $\E_{g_N}(\exp(z))=1$.
Assumption~\ref{ass: normal}~(ii) keeps the variance $\Var(z|\theta, N)$ constant across different values of $\t$, thus
making it easy to associate the $\IS^2$ asymptotic variances with $\s$.
Under Assumption~\ref{ass: normal}, the density $g_N(z|\t)$ depends only on $\s$
and we write it as $g(z|\s)$.
\begin{lemma}\label{lem:conditions}  
If Assumption 2 holds for a fixed $\s^2$, then  \eqref{eq: finite var} becomes
\beq\label{eq: stand IS conds}
\int h(\theta)^2 \left(\frac{\pi(\t)}{g_\IS(\t)}\right)^2 g_\IS(\t)d\t<\infty  .
\eeq
for both $h = \varphi$ and $h = 1$.
\end{lemma}
These are the standard conditions for IS \citep{Geweke:1989}.
The proof of this lemma is straightforward and omitted.

Recall that $\sigma^2_\IS(\varphi)/M$ and $\sigma^2_{\IS^2}(\varphi)/M$ are
respectively the asymptotic variances of the IS estimators
we would obtain when the likelihood is available
and when the likelihood is estimated.
We refer to the ratio $\sigma^2_{\IS^2}(\varphi)/\sigma^2_\IS(\varphi)$ as the inflation factor.
Theorem~\ref{the:ISefficency} obtains an expression for the inflation factor, shows that it is independent of $\varphi$, greater than or equal to 1 and increases exponentially with $\sigma^2$. Its proof is in the Appendix.

\begin{theorem}\label{the:ISefficency}  
Under Assumption~\ref{ass: normal} and the conditions in Theorem \ref{the:IS1},
\beq\label{e:IF_IS}
\frac{\sigma^2_{\IS^2}(\varphi)}{\sigma^2_\IS(\varphi)}=\exp(\sigma^2).
\eeq
\end{theorem}

\subsection{Optimally choosing the number of particles $N$}\label{Subsec:tradeoffIS}
In this section we wish to explicitly take account of both the statistical precision
of the estimator for the IS$^2$ method, given by \eqref{eq:Var_varphi} below, and
the computational effort. It is apparent that there will be a trade-off
between the two considerations. A large value of $N$ will result in a precise
estimator so that $\sigma^2$ will be small and the relative variance $\Var(\wh\varphi_{\IS^2})/\Var(\wh\varphi_{\IS})$ given by
\eqref{e:IF_IS} will be close to 1. However, the cost of such an estimator will be large due
to the large number of particles, $N$, required. Conversely, a small value of $N$ will
result in an estimator that is cheap to evaluate but results in a large value
of $\sigma^2$ and hence the variance of the IS$^2$ estimator relative to the IS estimator will be large. To explicitly
trade-off these considerations, we introduce the computational time measure $CT(\sigma^2)$
 which is a product of the relative variance and the computational effort.
Minimising  $CT(\sigma^2)$ results in an optimal value for
$\sigma^2$ and hence $N$.

Under Assumption~\ref{ass: normal}, $N=N_{\sigma^{2}}(\theta)=\gamma^{2}(\theta
)/\sigma^{2}$, so that the expected number of particles, over the draws of $\theta$,
is  ${\ov N}=\overline{{\gamma}^{2}}/\sigma^{2}$, where
$\overline {{\gamma}^{2}}:=\mathbb{E}_{g_{IS}}[\gamma^{2}(\theta)]$. We shall assume that the computational effort
required to compute $\widehat{p}_{N}(y|\theta)$ is $\tau_{0}+N\tau_{1}$,
where $\tau_0$ is the overhead cost of setting up the importance
sampler, and  $\tau_{1}$ is the cost of drawing each particle.  It is reasonable to assume that both $\tau_0$
and $\tau_1$ are independent of $\theta$. Hence, the expected
computational effort  for each $\theta$ is $\tau_{0}+\tau_{1}\overline{{\gamma}^{2}}/\sigma^{2}$.

From Theorems \ref{the:IS1} and \ref{the:ISefficency},
the variance of the estimator $\wh\varphi_{\IS^2}$ based on $M$
importance samples from $g_\IS(\t)$
is approximated by
\beq\label{eq:Var_varphi}
\Var(\wh\varphi_{\IS^2})\approx\frac{\sigma^2_{\IS^2}(\varphi)}{M} = \frac{\sigma^2_{\IS}(\varphi)}{M}\exp(\sigma^2)\approx \Var(\wh\varphi_{\IS})\exp(\sigma^2)
\eeq

We now define the measure of the computing time of the IS$^2$ method {\em relative} to
that of the hypothetical IS method, to achieve the same precision,
as
\begin{align}\label{eq: CT def}
CT(\sigma^{2}):=\exp(\sigma^{2})\left(  \tau_{0}+\tau_{1}\frac{\overline
{\gamma^{2}}}{\sigma^{2}}\right) ,
\end{align}
which is the product of the relative variance from \eqref{eq:Var_varphi} and the expected
computational effort.
It is straightforward to check that $\CT(\s^2)$ is convex and minimized at
\beq\label{eq:sigma_opt}
\s^2_\text{opt}=\frac{-\tau_1+\sqrt{\tau_1^2+4\tau_0\tau_1/\overline{\gamma^2}}}{2\tau_0/\overline{\gamma^2}}\;\text{if}\;\tau_0\not=0\;\text{and}\;
\s^2_\text{opt}=1\;\text{if}\;\tau_0=0.
\eeq
The optimal number of particles $N$ is such that
$\Var(z|\theta, N = N_{\sigma^2}(\theta))=\Var(\log\;\wh p_{N}(y|\theta,N = N_{\sigma^2}(\theta ))=\s^2_\text{opt}, $
for $\sigma^2 = \s^2_\text{opt}$.

Section~\ref{S: practical guide selecting N} discusses  practical guidelines for selecting an optimal number of particles.

\subsection{Optimal $N$ for estimating the marginal likelihood}\label{SS: optimal N marg likel}
The standard IS estimator of the marginal likelihood $p(y)$ when the likelihood is known is
$\wh p_\IS(y) = \sum_{i=1}^M w(\t_i)/M$,
with the weights $w(\t_i)$ from \eqref{e:standardISestimator}.
The corresponding marginal likelihood estimator when the likelihood is estimated
is $\wh p_{\IS^2}(y) = \sum_{i=1}^M {\wt w}(\t_i)/M$, with the weights
${\wt w}(\t_i)$ from \eqref{eq:weight_IS2}.
We have that $
\mathbb{E}_{g_{IS}}[\omega(\theta)]=\mathbb{E}_{\widetilde{g}_{IS}%
}[\widetilde{\omega}(\theta)]=p(y),
$
where again $\widetilde{\omega}(\theta)=e^{z}\omega(\theta)$. We again wish to
compare the ratios of the $IS$ estimator for the marginal likelihood with the
corresponding $IS^{2}$ estimator. Under Assumption \ref{ass: normal} that $z$ is Gaussian and
independent of $\theta$,
\begin{align*}
\mathbb{V}_{\widetilde{g}_{IS}}[\widetilde{\omega}(\theta)/p(y)]  &
=\mathbb{E}_{\widetilde{g}_{IS}}[e^{2z}\omega^{2}(\theta)/p(y)^{2}%
]-1=e^{\sigma^{2}}\mathbb{E}_{_{g_{IS}}}[\omega^{2}(\theta)/p(y)^{2}]-1\\
&  =e^{\sigma^{2}}\left(  \mathbb{V}_{_{g_{IS}}}[\omega(\theta
)/p(y)]+1\right)  -1.
\end{align*}
Consequently, the relative variance of the two schemes is given by%
\begin{align} \label{eq: ratio var marg lik}
\frac{\mathbb{V}_{\widetilde{g}_{IS}}[\widehat{p}_{IS^{2}}(y)]}{\mathbb{V}%
_{_{g_{IS}}}[\widehat{p}_{IS}(y)]} & =\frac{\mathbb{V}_{\widetilde{g}_{IS}%
}[\widetilde{\omega}(\theta)/p(y)]}{\mathbb{V}_{_{g_{IS}}}[\omega%
(\theta)/p(y)]}=\frac{e^{\sigma^{2}}(v+1)-1}{v},
\end{align}
where $v=\mathbb{V}_{_{g_{IS}}}[\omega(\theta)/p(y)]$ for notational
simplicity. For $v$ large, \eqref{eq: ratio var marg lik} shows that the resulting ratio is \eqref{eq:Var_varphi}.
Following the argument of section~\ref{Subsec:tradeoffIS}, the corresponding
computing time measure may be introduced as
\[
CT_{ML}(\sigma^{2}):=\left(  \tau_{0}+\tau_{1}\frac{\overline{\gamma}^{2}%
}{\sigma^{2}}\right)  \left\{  \frac{e^{\sigma^{2}}(v+1)-1}{v}\right\}  ,
\]
where the subscript $ML$ indicates this is for the marginal likelihood estimator.

Let $\s^2_{\min}(v)$ minimize $CT_{ML}(\sigma^{2})$ for a given $v$.
The following proposition summarizes some properties of $\s^2_{\min}(v)$, where
$\s^2_\text{opt}$ below is given by \eqref{eq:sigma_opt}.
Its proof is in section~\ref{app supp: proof} of the supplementary material.
\begin{proposition}\label{proposition:marg_lik}  
\begin{enumerate}
\item[(i)] For any value of $v$, $\CT_{ML}(\s^2)$ is a convex function of $\sigma^{2}$;
therefore $\s^2_{\min}(v)$ is unique.
\item[(ii)] $\s^2_{\min}(v)$ increases as $v$ increases and
$\s^2_{\min}(v)\longrightarrow\s^2_\text{opt}$ in \eqref{eq:sigma_opt} as $v\longrightarrow\infty$.
\end{enumerate}
\end{proposition}
It can be readily checked that $\s^2_{\min}(v)$ is insensitive to $v$ in the sense that
it is a flat function of $v$ for large $v$, with $\sigma^2_{\min}\to\s^2_\text{opt}=0.17$ as $v $ increases.
See also Table~\ref{tab:marg_lik} in section \ref{SS: optimal N marginal likel} of the supplementary material. Based on these observations,
we advocate using $\sigma^2_\text{opt}$ in \eqref{eq:sigma_opt} as the optimal value
of the variance of the log likelihood estimates in estimating the marginal likelihood.

\subsection{Constructing the importance density $g_{\IS}(\theta)$ }\label{Subsec:proposal}
A potential drawback with $\IS^2$, when estimating the posterior of $\theta$, {\em but not} for estimating the marginal likelihood,
is that its performance depends on
the proposal density $g_\IS(\t)$ for $\theta $, which  may
be difficult to obtain in complex models.
This section outlines the Mixture of $t$ by Importance Sampling Weighted Expectation Maximization (MitISEM) of \cite{hod2012}
approach we used in the article for
designing efficient and reliable proposal densities for $\IS^2$.
Section~\ref{SS_supp:proposal} of the supplementary material outlines a
second approach based on annealed importance sampling for models with latent variables.

MitISEM constructs a mixture of $t$ densities for approximating the target distribution
by minimizing the Kullback–-Leibler divergence between the target and the $t$ mixture,
and can handle target distributions that have non-standard shapes such as multimodality and skewness.
When the likelihood is available, the MitISEM can be used to design proposal densities
that approximate the posterior accurately.
It is natural to use an estimated likelihood when the likelihood is  unavailable.
We write $\wh p_N(y|\t)=\wh p_N(y|\t,u)$, with $u$ a fixed random number stream for all $\theta$.
The target distribution in the MitISEM is  $p(\t)\wh p_N(y|\t,u)/p(y)$, which can be considered
as the posterior $p(\theta|y,u)$ conditional on $y$ and the common random numbers $u$.
Our procedure is analogous to using common random numbers $u$ to obtain simulated maximum likelihood estimates of
$\theta$ \citep[see, e.g.,][]{Gourieroux:1995}, except that we obtain a histogram estimate of the \lq posterior\rq{}
$p (\theta| u, y) \propto p(y|\theta, u ) p(u) $. This \lq posterior\rq{}  is biased, but sufficiently good to obtain a good proposal density.

\section{Practical guidelines for selecting an optimal $N$} \label{S: practical guide selecting N}
Let $\Var_{N,\t}(z) = \Var(z|N, \theta ) $ and let $\wh\Var_{N,\t}(z)$ be an estimate of $\Var_{N,\t}(z)$.
We suggest the following practical guidelines for tuning the optimal number of particles $N$.
Section \ref{sec:IS_theory} of the supplementary material gives another strategy.
Although $N$ generally depends on $\t$, we suppress  this dependence for notational simplicity.

From \eqref{eq:sigma_opt}, if $\tau_0=0$, then it is necessary to tune $N$ such that $\wh\Var_{N,\t}(z)=1$.
A simple strategy is to start with some small $N$ and increase it if $\wh\Var_{N,\t}(z) > 1$.

If $\tau_0>0$, we first need to estimate $\overline{\gamma^2}=\E_{g_\IS}(\gamma^2(\t))$.
Let $\{\t^{(1)},...,\t^{(J)}\}$ be $J$ draws from the importance density $g_\IS(\t)$.
Then, starting with some large $N_0$ , $\overline{\gamma^2}$ can be estimated by
\beq\label{eq:gamma^2}
\wh{\overline{\gamma^2}} := \frac{1}{J}\sum_{j=1}^J\wh{\gamma^2}(\t^{(j)})=\frac{N_0}{J}\sum_{j=1}^J\wh\Var_{N_0,\t^{(j)}}(z),
\eeq
as $\wh\Var_{N_0,\t^{(j)}}(z)=\wh\gamma^2(\t^{(j)})/N_0$.
By substituting $\wh{\overline{\gamma^2}} $ into \eqref{eq:sigma_opt},  we obtain an estimate $\wh{\s}^2_\text{opt}$ of ${\s}^2_\text{opt}$.
Now, for each draw of $\t_i$ in Algorithm 1, we start with some small $N$ and increase $N$ if $\wh\Var_{N,\t_i}(z)>\wh{\s}^2_\text{opt}$.
Note that the new draw $\t_i$ with the estimate $\wh\Var_{N,\t_i}(z)$ can be added to \eqref{eq:gamma^2} to improve the estimate $\wh{\overline{\gamma^2}}$.

In the applications in Section~\ref{Sec:Examples}, we use the \textit{time
normalized variance} (TNV) of an $\text{\textrm{IS}}^{2}$ estimator
$\widehat{\varphi}_{\text{\textrm{IS}}^{2}}$ as a measure of its inefficiency, which we define
as
\begin{align}\label{eq:TNV}%
\mathrm{TNV}(M,\ov N)&:=\mathbb{V}(\widehat{\varphi}_{\text{\textrm{IS}}^{2}}%
)\times (M\tau_0 + \ov N M \tau_1)
\end{align}
where $\overline{N}=(N_{1}+\cdots+N_{M})/M$ and   $(M\tau_0 + \ov N M \tau_1)$ is the total computing time of an
$\text{\textrm{IS}}^{2}$ run with $M$ importance samples $\{\theta
_{i},\,i=1,\dots,M\}$, with $\theta$ generated by $g_{IS}$, and using $N_{i}$
particles to estimate the likelihood at the value $\theta_{i}$. This is
essentially the empirical analogue of the computing time measure
$CT(\sigma^{2})$ of~\eqref{eq: CT def}, where the actual rather than the expected
computational resources are calculated. The correspondence may be seen by
noting from Theorem \ref{the:ISefficency} that $\mathbb{V}(\widehat{\varphi
}_{\text{\textrm{IS}}^{2}})=\mathbb{V}(\widehat{\varphi}_{\text{\textrm{IS}}%
})\exp(\sigma^{2})$ and that, for large $M$, the computational resources
component of~\eqref{eq: CT def},%
\[
\tau_{0}+(\tau_{1}/\sigma^{2})E_{g_{IS}}(\gamma^{2})\approx\tau_{0}+(\tau
_{1}/\sigma^{2})\frac{1}{M}\sum_{i=1}^{M}\gamma^{2}(\theta_{i})\approx\tau
_{0}+\tau_{1}\overline{N}.
\]
The last approximation arises as ${\overline{N}}=\mathbb{E}_{g_{IS}}%
[\gamma^{2}(\theta)]/\sigma^{2}=\overline{\gamma}^{2}/\sigma^{2},$ see Section~\ref{Subsec:tradeoffIS}
Hence,%
\[
\mathrm{TNV}(M,\ov N)\approx M\mathbb{V}(\widehat{\varphi}_{\text{\textrm{IS}}%
})CT(\sigma^{2}),
\]
where $\mathbb{V}(\widehat{\varphi}_{\text{\textrm{IS}}})$ is the variance of
the standard importance sampler and so is independent of $\overline{N}$ and
$\sigma^{2}$.

\section{Applications}\label{Sec:Examples}
\subsection{The mixed logit and generalized multinomial logit models}\label{sec:gmnlmodel}
\subsubsection{Model and data}

The generalized multinomial logit (GMNL) model of \cite{fklw2010} specifies
the probability of individual $i$ choosing alternative $j$ on occasion $t$ as
\begin{equation}\label{eq:choiceprob}
p(\textrm{$i$ chooses $j$ at $t$}|X_{it},\beta_{i})=\frac{\exp(\beta_{0ij}+\sum_{k=1}^{K}\beta_{ki}x_{kijt})}{\sum_{h=1}^{J}\exp(\beta_{0ih}+\sum_{k=1}^{K}\beta_{ki}x_{kiht})},
\end{equation}
where $\beta_{i}=(\beta_{0i1},\ldots,\beta_{0iJ},\beta_{1i},\ldots,\beta_{Ki})'$ and $X_{it}=(x_{1i1t},\ldots,x_{Ki1t},\ldots,x_{1iJt},\ldots,x_{KiJt})'$ are the vectors of utility weights and choice attributes respectively.
The GMNL model specifies the alternative specific constants as
$\beta_{0ij}=\beta_{0j}+\eta_{0ij}$ with $\eta_{0ij}\sim \Nn(0,\sigma_{0j}^2)$
and the attribute weights as
\beqn
\beta_{ki}=\lambda_i\beta_k+\gamma\eta_{ki}+(1-\gamma)\lambda_i\eta_{ki},\qquad \lambda_i=\exp(-\delta^2/2+\delta\zeta_{i}), \qquad k=1,\ldots,K,
\eeqn
with $\eta_{ki}\sim \Nn(0,\sigma_k^2)$ and  $\zeta_{i}\sim \Nn(0,1)$.
The expected value of the scaling coefficients $\lambda_i$ is one, implying that $\Ee(\beta_{ki})=\beta_k$.

When $\delta=0$ (so that $\lambda_i=1$ for all individuals) the GMNL model reduces to the mixed logit (MIXL) model, which we also consider in our analysis.  The MIXL model captures heterogeneity in consumer preferences by allowing individuals to weight the choice attributes differently. By introducing taste heterogeneity, the MIXL specification avoids the restrictive independence of irrelevant alternatives property of the standard multinomial logit model \citep{fklw2010}.
The GMNL model additionally allows for scale heterogeneity through the random variable $\lambda_i$, which changes all attribute weights
simultaneously. Choice behavior in this model can therefore be more random for some consumers than others.  The $\gamma$ parameter weights the specification between two alternative ways of introducing scale heterogeneity into the model.

The parameter vector is $\theta=(\beta_{01},\ldots \beta_{0J},\sigma_{0}^2,\beta_1,\ldots,\beta_K,\sigma_1^2,\ldots,\sigma_K^2,\delta^2,\gamma)'$, while the vector of latent variables for each individual is $x_{i}=(\eta_{0i},\ldots,\eta_{Ki},\lambda_i)$. The likelihood is therefore
\begin{align}\label{eq:gmnllik}
p(y|\theta)=\prod_{i=1}^{I}p(y_i|\theta)=\prod_{i=1}^{I}\left[\int  \(\prod_{t=1}^{T} p(y_{it}|x_{i})\) p(x_i) \dd x_i\right],
\end{align}
where $y_{it}$ is the observed choice, $y=(y_{11},\ldots,y_{1T},\ldots,y_{I1},\ldots,y_{IT})'$ and $p(y_{it}|x_{i})$ is given by the choice probability \eqref{eq:choiceprob}.

 We apply our methods to the Pap smear data set considered  by \cite{fklw2010}, who used simulated maximum likelihood estimation. In this data set, $I=79$ women choose whether or not to have a Pap smear test ($J=2$) on $T=32$ choice scenarios. We let the observed choice for individual $i$ at occasion $t$ be $y_{it}=1$ if the woman chooses to take the test and $y_{it}=0$ otherwise.
 Table \ref{tab:choices} lists the choice attributes and the associated coefficients.  We set $\sigma_5^2=0$ in our
 analysis since we found no evidence of heterogeneity for this attribute beyond the scaling effect.
The model \eqref{eq:choiceprob} is identified
by setting the coefficients as zero when not taking the test.

\begin{table}[h]
\caption{Choice attributes for the pap smear data set}\label{tab:choices}
\begin{center}
\begin{tabular}{lll}
\hline\hline
Choice attributes & Values & Associated parameters \\
\hline
Alternative specific constant for test & 1 & $\beta_{01}, \sigma_{01}$ \\
Whether patient knows doctor &  0 (no), 1 (yes) & $\beta_1, \sigma_1$ \\
Whether doctor is male &  0 (no), 1 (yes) & $\beta_2, \sigma_2$ \\
Whether test is due &  0 (no), 1 (yes) & $\beta_3, \sigma_3$ \\
Whether doctor recommends test &  0 (no), 1 (yes) & $\beta_4, \sigma_4$ \\
Test cost & \{0, 10, 20, 30\}/10 & $\beta_5$ \\
\hline\hline
\end{tabular}
\end{center}
\end{table}

We specify the priors as
$\beta_{01}\sim \Nn(0,100)$, $\sigma_{01} \propto (1+\sigma_{01}^2)^{-1}$,
$\beta_k\sim \Nn(0,100)$, $\sigma_k \propto (1+\sigma_{k}^2)^{-1}$, for $k=1,\ldots,K$,
$\delta\propto (1+\delta/0.2)^{-1}$,
and $\gamma \sim \textrm{U}(0,1)$.
The standard deviation parameters have half-Cauchy priors as suggested by \cite{Gelman2006}.

\subsubsection{Implementation details}\label{sec:isdetails}

We estimate the likelihood \eqref{eq:gmnllik} by integrating out the vector of latent variables for each individual separately using different approaches for the MIXL and GMNL models. For the MIXL model, we combine the efficient importance sampling (EIS) method of \cite{ZR2007} with the defensive sampling approach of \cite{Hesterberg1995}. The importance density is the two component defensive mixture
\[h(x_i|y_{i1},\ldots,y_{iT})=\pi h^\text{EIS}(x_i|y_{i1},\ldots,y_{iT})+(1-\pi)p(x_i),\]
where $h^\text{EIS}(x_i|y_{i1},\ldots,y_{iT})$ is a multivariate Gaussian importance density obtained using EIS. Following
 \cite{Hesterberg1995}, including the natural sampler $p(x_i)$ in the mixture ensures that the importance weights are bounded. We set the mixture weight as $\pi=0.5$. For the GMNL model, we follow \cite{fklw2010} and use the model density $p(x_i)$ as an importance sampler. We implement this simpler approach for the GMNL model because the occurrence of large values of $\lambda_i$ causes the defensive mixture estimates of the log-likelihood to be pronouncedly right skewed in this case.

The panel structure of the problem implies that the log-likelihood estimates are sums of independent estimates $\log\hspace{1mm}  \widehat{p}(y_i|\theta)$ for each individual. In order to target a certain variance $\sigma^2$ for the log-likelihood estimator, we choose the number of particles for each individual ($N_i$) and parameter combination $\theta$ such that $\Var(\log\hspace{1mm} \widehat{p}(y_i|x_i))\approx \sigma^2/I$.  We implement this scheme by using a certain number of initial importance samples and the jackknife method to estimate $\gamma^2_i(\theta)$, the asymptotic variance of $\log \hspace{1mm} \widehat{p}(y_i|x_i)$, and select $N_i(\theta)=\widehat{\gamma^2_i}(\theta)I/\sigma^2$. The preliminary number of particles is $N=20$ for the MIXL model and $N=2,500$ for the GMNL model.

To obtain the parameter proposals $g_\IS(\t)$, we use the MitISEM approach as described in section~\ref{Subsec:proposal},
which approximates the posterior of the two models as two component mixtures of multivariate Student's $t$ distributions.

We implement two standard variance reduction methods at each IS stage: stratified mixture sampling and antithetic sampling. The first consists of sampling from each component at the exact proportion of the mixture weights. For example, when estimating the likelihood for the MIXL model we generate exactly $\pi N$ draws from the efficient importance density $h^\text{EIS}(x_i|y_{i1},\ldots,y_{iT})$ and $(1-\pi) N$  draws from $p(x_i)$. The antithetic sampling method consists of generating pairs of perfectly negatively correlated draws from each mixture component, see, e.g., \cite{Ripley87}.

\subsubsection{Posterior analysis}

Table \ref{tab:posterior} presents selected posterior statistics estimated by the IS$^2$ method under the optimal schemes (targeting $\sigma^2_\text{opt}=0.17$ for the MIXL model and $\sigma^2_\text{opt}=1$ for the GMNL model). We estimated the posterior distribution using $M=50,000$ importance samples for the parameters.  We also estimated the Monte Carlo standard errors by bootstrapping the importance samples. We highlight two results. The MC standard errors are low in all cases, illustrating the high efficiency of the IS$^2$ method for carrying out Bayesian inference for these two models. The estimates of the log of the marginal likelihood allow us to calculate a Bayes factor of approximately 20 for the GMNL model over the MIXL model, so that there is some evidence to suggest the presence of scale heterogeneity for this dataset.

\subsubsection{Comparing IS$^2$, pseudo marginal Metropolis-Hastings and SMC$^2$}\label{SS: comparison PMMH IS2 and pseudo marg}

The IS$^2$, PMMH, and SMC$^2$ methods can tackle the same set of problems. This section compares the performance of IS$^2$, PMMH and SMC$^2$. We evaluate the first two methods  under the same independent proposal for the parameter vector $\theta$ and likelihood estimation method, which leads to a direct comparison. To accurately measure the Monte Carlo efficiency of the three methods for the MIXL model, we run 500 independent replications of the IS$^2$ algorithm using $M=5,000$ importance samples, 500 independent PMMH Markov chains with 5,000 iterations (after a burn-in of 1,000 iterations), and 500 independent replications of the SMC$^2$ method with 5,000 particles. We focus on estimates of the posterior mean. We obtain the initial draw for the PMMH Markov chain by using the discrete IS$^2$ approximation of the posterior. We implement the SMC$^2$ method by sequentially updating the posterior for each individual. We estimate the likelihood contribution from each individual in the same way as for the IS$^2$ for the PMMH method. We initialise the method implementing IS$^2$ for the first patient in the dataset. Due to the large prior variances, we found it necessary to fit a \cite{hod2012} proposal for this step. For the sequential updates, we follow \cite{Chopin:2012} and use independent Metropolis-Hastings rejuvenation steps by fitting a multivariate Gaussian proposal to the current particle system.

 In addition, we consider two improvements to the basic IS$^2$ method: the Pareto smoothing (PS) method of \cite{Vehtari2015} and  Quasi-Monte Carlo (QMC) sampling of the parameters and trimmed means. The motivation for using the
Pareto smoothing method is to stabilize the right tail of the distribution of IS$^2$ weights, increasing the robustness of the method when there is noise in the weights
because the likelihood is estimated. We implement QMC by replacing the standard normal draws used to generate samples from $q(\theta|y)$ by scrambled Sobol sequences converted to quasi-random $N(0,1)$ draws through the inverse CDF method.

Table \ref{tab:comparison1} presents the MC mean squared errors for the posterior mean estimates and the MC standard error estimates, normalised by the performance of the PMMH method for the MIXL model.  We approximate the posterior mean by combining the particles generated in all replications for the standard IS$^2$ method. The results support the argument that the IS$^2$ method leads to improved efficiency due to the use of independent samples and variance reduction methods. The table shows that the standard IS$^{2}$ method leads to approximately 70\% average reductions in the MC MSE compared to the PMMH method and the SMC$^2$ implementation.  We find that Pareto smoothing
dramatically increases the accuracy of the MC standard error estimates, although it does not seem to have a large impact on the MSE of the posterior mean estimates for this model. In results not reported in the table, we found that Pareto smoothing leads to better behaved, near Gaussian,
posterior mean estimates, so that we recommend this approach in practice to ensure robustness. The use of QMC leads to a modest incremental improvement of approximately 15\% in the MSE of the posterior mean estimates. Regarding the computational costs, the IS$^2$ and PMMH methods are equivalent in our setting, except for the additional cost of the burn-in period for the latter. Each replication of IS$^2$ and PMMH required approximately 50 seconds in a machine equipped with an Intel Core i7-4790 CPU with 4 cores. In contrast, each SMC$^2$ replication required approximately six minutes, highlighting the larger computational efficiency of IS$^2$ and PMMH for off-line estimation.

\begin{table}[!htbp]
\begin{center}
\caption{IS$^2$, PMMH and SMC$^2$ comparison for the MIXL model for the Monte Carlo MSE and the variance of the Monte Carlo standard error estimates.
IS$^2$/PS means IS$^2$ with Pareto smoothing and IS$^2$/PS/QMC means IS$^2$/PS with quasi Monte Carlo.}\label{tab:comparison1}
\begin{threeparttable}
{\footnotesize }
\begin{tabular}{lccccccccc}
\hline\hline
& \multicolumn{5}{c}{Monte Carlo mean-square error} &  & \multicolumn{3}{c}{Variance of MC SE estimate}\\
 & IS$^2$ & IS$^2$/PS & IS$^2$/PS/QMC &SMC$^2$ & PMMH & \multicolumn{1}{l}{} & IS$^2$ & IS$^2$/PS & PMMH \\
 \cline{2-6}\cline{8-10}
$\beta_{01}$  & 0.298 & 0.269 & 0.206 & 2.662 & 1.000 &  & 0.191 & 0.003 & 1.000 \\
$\beta_1$  &0.338 & 0.300 & 0.222 & 0.722 & 1.000 &  & 0.165 & 0.002 & 1.000 \\
$\beta_2$ & 0.335 & 0.290 & 0.209 & 0.693 & 1.000 &  & 0.173 & 0.002 & 1.000 \\
$\beta_3$ & 0.314 & 0.292 & 0.243 & 0.821 & 1.000 &  & 0.181 & 0.003 & 1.000 \\
$\beta_4$ & 0.253 & 0.248 & 0.240 & 0.983 & 1.000 &  & 0.133 & 0.002 & 1.000 \\
$\beta_5$ & 0.292 & 0.276 & 0.262 & 0.901 & 1.000 &  & 0.147 & 0.001 & 1.000 \\
$\sigma_0$ &0.320 & 0.276 & 0.215 & 1.168 & 1.000 &  & 0.174 & 0.003 & 1.000 \\
$\sigma_1$ & 0.324 & 0.304 & 0.218 & 0.728 & 1.000 &  & 0.177 & 0.002 & 1.000 \\
$\sigma_2$ & 0.310 & 0.285 & 0.294 & 0.615 & 1.000 &  & 0.193 & 0.002 & 1.000 \\
$\sigma_3$ & 0.299 & 0.276 & 0.252 & 0.599 & 1.000 &  & 0.166 & 0.002 & 1.000 \\
$\sigma_4$ & 0.161 & 0.095 & 0.074 & 0.313 & 1.000 &  & 0.387 & 0.003 & 1.000 \\
\hline\hline
\end{tabular}
\end{threeparttable}
\end{center}
\end{table}

Table \ref{tab:comparison2} extends the analysis to the more challenging GMNL specification. We consider only the IS$^2$ and PMMH methods in this case, as the the SMC$^2$ method was prohibitively costly. The results are similar to the ones in Table \ref{tab:comparison2}: the IS$^2$ method leads to 78\% reductions in MSE on average compared to PMMH. When considering only the variations of the $\IS^2$ method, the results show that for this model the use of Pareto smoothing substantially improves the performance of the method, leading to 57\% lower MSEs compared to the standard implementation. As before, the results also suggest that the estimates of the MC variance within each of the replications are substantially more accurate under IS$^2$, especially with Pareto smoothing.

\begin{table}[!tbp]
\begin{center}
\caption{IS$^2$ and PMMH comparison for the GMNL model. IS$^2$/PS means IS$^2$ with Pareto smoothing and IS$^2$/PS/QMC means IS$^2$/PS with quasi Monte Carlo.}\label{tab:comparison2}
\begin{threeparttable}
\begin{tabular}{lcccccccc}
\hline\hline
& \multicolumn{4}{c}{Monte Carlo mean-square error} &  & \multicolumn{3}{c}{Variance of MC SE estimate}\\
 & IS$^2$ & IS$^2$/PS & IS$^2$/PS/QMC & PMMH & \multicolumn{1}{l}{} & IS$^2$ & IS$^2$/PS & PMMH \\
 \cline{2-5}\cline{7-9}
$\beta_{01}$ & 0.254 & 0.091 & 0.080 & 1.000 &  & 0.196 & 0.006 & 1.000 \\
$\beta_1$ & 0.350 & 0.096 & 0.077 & 1.000 &  & 0.420 & 0.022 & 1.000 \\
$\beta_2$ & 0.165 & 0.065 & 0.062 & 1.000 &  & 0.068 & 0.007 & 1.000 \\
$\beta_3$ & 0.335 & 0.102 & 0.065 & 1.000 &  & 0.234 & 0.014 & 1.000 \\
$\beta_4$ & 0.199 & 0.102 & 0.075 & 1.000 &  & 0.096 & 0.013 & 1.000 \\
$\beta_5$ & 0.324 & 0.211 & 0.152 & 1.000 &  & 0.126 & 0.019 & 1.000 \\
$\sigma_0$ & 0.215 & 0.119 & 0.107 & 1.000 &  & 0.093 & 0.008 & 1.000 \\
$\sigma_1$ & 0.110 & 0.059 & 0.051 & 1.000 &  & 0.058 & 0.014 & 1.000 \\
$\sigma_2$ & 0.209 & 0.088 & 0.062 & 1.000 &  & 0.124 & 0.014 & 1.000 \\
$\sigma_3$ & 0.196 & 0.122 & 0.097 & 1.000 &  & 0.105 & 0.012 & 1.000 \\
$\sigma_4$ & 0.188 & 0.058 & 0.039 & 1.000 &  & 0.125 & 0.005 & 1.000 \\
$\delta$ & 0.110 & 0.028 & 0.018 & 1.000 &  & 0.077 & 0.003 & 1.000 \\
$\gamma$ & 0.239 & 0.094 & 0.049 & 1.000 &  & 0.250 & 0.016 & 1.000\\
\hline\hline
\end{tabular}
\end{threeparttable}
\end{center}
\end{table}

\subsection{Stochastic volatility}\label{sec:sv}
\subsubsection{Model, data and importance sampling}

We consider the univariate two-factor stochastic volatility (SV) model with leverage effects
\begin{align*}
&y_t=\exp([c+x_{1,t}+x_{2,t}]/2)
\eps_t,\quad t=1,\ldots,n, \qquad
 x_{i,t+1}=\phi_1x_{i,t}+\rho_1\sigma_{i,\eta}\eps_t+\sqrt{1-\rho_i^2}\sigma_{i,\eta}\eta_{i,t},\quad i=1,2,
\end{align*}
where the return innovations are i.i.d. and have standardized Student's $t$ distributions with $\nu$ degrees of freedom and $1>\phi_1>\phi_2>-1$ for stationarity and identification. This SV specification incorporates the most important empirical features of a
volatility series, allowing for fat-tailed return innovations, a negative correlation between lagged returns and volatility (leverage effects) and quasi-long memory dynamics through the superposition of AR(1) processes \citep[see, e.g.,][]{bNnS2002}. The two volatility factors have empirical interpretations as long term and short term volatility components.

We use the following priors:~$\nu \sim \textrm{U}(2,100)$, $c\sim \Nn(0,1)$, $\phi_1\sim\textrm{U}(0,1)$,
$\sigma_1^2 \sim \textrm{IG}(2.5, 0.0075)$, $\rho_1 \sim \textrm{U}(-1,0)$, $\phi_2\sim\textrm{U}(0,\phi_1)$,
$\sigma_2^2 \sim \textrm{IG}(2.5,0.045)$, and $\rho_2 \sim \textrm{U}(-1,0)$,
and estimate the two factor SV model using daily returns from the S\&P 500 index obtained from the CRSP database. The sample period runs from January 1990 to December 2012, for a total of $5,797$ observations.
As in the previous example, we obtain the parameter proposal $g_\IS(\theta)$ using the MitISEM method of \cite{hod2012},
where we allow for two mixture components.
To estimate the likelihood, we apply the particle efficient importance sampling (PEIS)
 method of \cite{PEIS}, which is a sequential Monte Carlo procedure for likelihood estimation that uses the EIS algorithm of \cite{ZR2007} to obtain a high-dimensional importance density for the states. We implement PEIS using antithetic variables as described in \cite{PEIS}.

\subsubsection{Optimal number of particles}

To investigate the accuracy of the likelihood estimates for the stochastic volatility model, we replicate the Monte Carlo exercise of section~\ref{sec:likeval1} for this example. Table \ref{tab:lik2} displays the average variance, skewness and kurtosis of the log-likelihood estimates based on $N=10$ and $N=20$ particles, as well as the computing times.  We base the analysis on 1,000 draws from the proposal for $\theta$ and 200 independent likelihood estimates for each parameter value.

The results show that the PEIS method is highly accurate for estimating the likelihood of the SV model, despite the large number of observations. When $N=10$, the inflation factor of \eqref{e:IF_IS} is approximately 1.01, indicating that performing IS with the PEIS likelihood estimate has essentially the same efficiency as if the likelihood was known. On the other hand, the log-likelihood estimates display positive skewness and excess kurtosis and are clearly non-Gaussian, suggesting that the theory of
section~\ref{Subsec:tradeoffIS} only holds approximately for this example.

We approximate the optimal number of particles by assuming that the variance of log-likelihood estimates is constant across different values of $\theta$, as it is not possible to use the jackknife method for estimating the variance of log-likelihood estimators based on particle methods. Based on the average variance of the log-likelihood estimate with $N=20$, we estimate the asymptotic variance of the log-likelihood estimate to be $\widehat{\gamma^2}(\theta)=0.1$ on average. The last row of the table allows us to calculate that the overhead of estimating the likelihood is $\tau_0=1.051$ seconds, while the computational cost of each particle is $\tau_1=0.018\times10^{-1}$ seconds.  Assuming normality of the log-likelihood estimate and that $\gamma^2(\theta)=\bar{\gamma}^2=0.1$ for all $\theta$, the optimal number of particles is $N_{\textrm{opt}}=8$.

Figure \ref{fig:svoptimal} plots the predicted relative time normalized variance $\text{TNV}(M,N)/\text{TNV}(M,N_\text{opt})$
as a function of $N$. The figure suggests that the efficiency of $\IS^2$ for this example is fairly insensitive to the number of particles at the displayed range $N\leq40$, with even the minimal number of particles $N=2$ (including an antithetic draw) leading to a highly efficient procedure.

\begin{table}[!tbp]
\begin{center}
\caption{Stochastic volatility - log-likelihood evaluation.}\label{tab:lik2}
\begin{threeparttable}
\begin{tabular}{lcc}
\hline\hline
 & $N=10$ & $N=20$ \\
 \cline{2-3}
Variance & 0.009 & 0.005 \\
Skewness & 0.485 & 0.394 \\
Kurtosis & 3.796 & 3.640 \\
Time (s) & 1.069 & 1.087 \\
\hline\hline
\end{tabular}
\end{threeparttable}
\end{center}
\end{table}

\begin{figure}
    \begin{center}
        \includegraphics[scale=0.45]{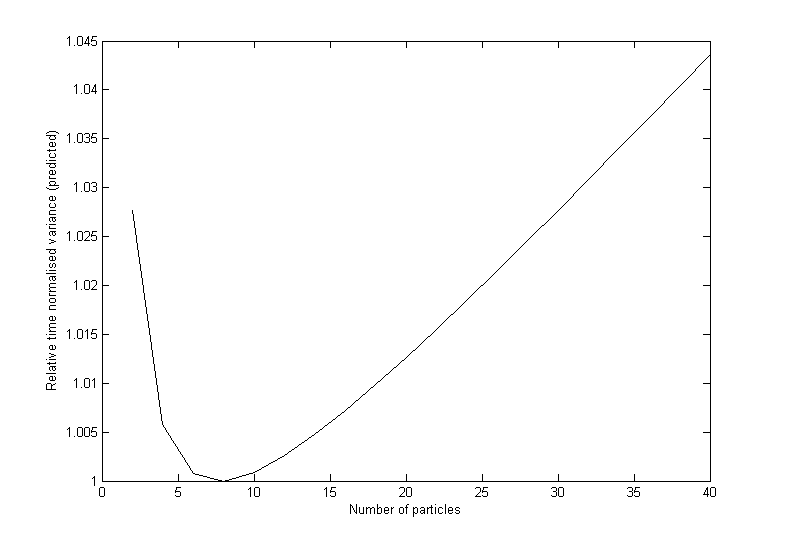}
    \end{center}
    \caption{Predicted relative time normalized variance against the number of particles for the stochastic volatility model.}\label{fig:svoptimal}
\end{figure}

Table \ref{tab:example2} displays estimates of the relative variances  for the posterior means based on $M=50,000$ importance samples for the parameters. We estimate the Monte Carlo variance of the posterior statistics by bootstrapping the importance samples. The results are consistent  with Figure~\ref{fig:svoptimal} and indicate that the efficiency of the IS$^2$ method is insensitive to the number of particles in this example, with no value of $N$ emerging as clearly optimal in the table. Based on this result,
we recommend a conservative number of particles in practice as there is no important efficiency cost in choosing $N$ moderately above the optimal theoretical value.

\begin{table}[!tbp]
\begin{center}
\caption{Stochastic volatility: relative variances for posterior inference.
The table shows the relative variances for IS$^2$ for different numbers of importance samples used for estimating the likelihood.}\label{tab:example2}
\begin{threeparttable}
\begin{tabular}{lcccccc}
\\
\hline\hline
 & \multicolumn{1}{c}{$N=2$} & \multicolumn{1}{c}{$N=4$} & \multicolumn{1}{c}{$N=8$} & \multicolumn{1}{c}{$N=12$} & \multicolumn{1}{c}{$N=16$} & \multicolumn{1}{c}{$N=20$} \\
\cline{2-7}
$\nu$ & 0.963 & 0.873 & 1.000 & 0.937 & 0.893 & 0.885 \\
$c$ & 0.901 & 0.978 & 1.000 & 0.996 & 0.945 & 0.966 \\
$\phi_1$ & 1.007 & 1.019 & 1.000 & 0.991 & 0.980 & 0.996 \\
$\sigma_1^2$ & 0.995 & 1.000 & 1.000 & 1.022 & 0.974 & 1.032 \\
$\rho_1$ & 1.086 & 1.004 & 1.000 & 1.034 & 0.994 & 1.015 \\
$\phi_2$ & 0.952 & 0.947 & 1.000 & 0.999 & 0.991 & 1.007 \\
$\sigma_2^2$ & 1.016 & 1.030 & 1.000 & 1.005 & 0.969 & 0.972 \\
$\rho_2$ & 1.039 & 1.058 & 1.000 & 1.025 & 1.028 & 1.011 \\
\\
Average & 0.995 & 0.989 & 1.000 & 1.001 & 0.972 & 0.986 \\
Theoretical & 1.038 & 1.013 & 1 & 0.996 & 0.994 & 0.993 \\
\\
Time (minutes) & 258 & 260 & 262 & 263 & 264 & 265 \\
Estimated TNV & 0.977 & 0.978 & 1.000 & 1.005 & 0.978 & 0.994 \\
Theoretical & 1.019 & 1.002 & 1.000 & 0.999 & 1.001 & 1.001 \\
\hline\hline
\end{tabular}
\end{threeparttable}
\end{center}
\end{table}

\subsubsection{Posterior analysis}

Table \ref{tab:posterior2} presents estimates of selected posterior distribution statistics estimated by the IS$^2$ method. We estimated the posterior distribution using $M=50,000$ importance samples for the parameters and $N=20$ particles to estimate the likelihood. As before, we estimate the Monte Carlo standard errors by bootstrapping the importance samples. The results show that IS$^2$ leads to highly accurate estimates of the posterior statistics and
 that the short term volatility component is almost entirely driven by leverage effects.

\begin{table}[!tbp]
\begin{center}
\caption{Stochastic Volatility: posterior statistics.
The table presents estimates of selected posterior distribution statistics based on $M=50,000$ importance samples for the parameters. The Monte Carlo standard errors are in brackets. $\widehat{\ESS}_{\IS}$ and $\widehat{\ESS}_{\IS^2}$  are the estimates of the equivalent sample sizse of IS and IS$^2$; see table~\ref{tab:posterior}.  }\label{tab:posterior2}
\begin{threeparttable}
\begin{tabular}{lcccccc}
\\
\hline\hline
& Mean & Std. Dev. & Skew. & Kurt. & \multicolumn{2}{c}{90\% Credible Interval} \\
\cline{2-7}
$\nu$ & $\underset{[0.016]}{13.666}$ & $\underset{[0.024]}{2.619}$ & $\underset{[0.036]}{1.097}$ & $\underset{[0.173]}{5.026}$ & $\underset{[0.019]}{10.203}$ & $\underset{[0.067]}{18.543}$ \\[10pt]
$c$ & $\underset{[0.001]}{0.044}$ & $\underset{[0.001]}{0.178}$ & $\underset{[0.034]}{0.311}$ & $\underset{[0.103]}{3.652}$ & $\underset{[0.002]}{-0.233}$ & $\underset{[0.004]}{0.346}$ \\[10pt]
$\phi_1$ & $\underset{[0.000]}{0.994}$ & $\underset{[0.000]}{0.002}$ & $\underset{[0.015]}{-0.387}$ & $\underset{[0.041]}{3.152}$ & $\underset{[0.000]}{0.992}$ & $\underset{[0.000]}{0.997}$ \\[10pt]
$\sigma_1^2$ & $\underset{[0.000]}{0.007}$ & $\underset{[0.000]}{0.002}$ & $\underset{[0.016]}{0.654}$ & $\underset{[0.061]}{3.653}$ & $\underset{[0.000]}{0.004}$ & $\underset{[0.000]}{0.010}$ \\[10pt]
$\rho_1$ & $\underset{[0.001]}{-0.490}$ & $\underset{[0.001]}{0.106}$ & $\underset{[0.014]}{0.422}$ & $\underset{[0.031]}{2.852}$ & $\underset{[0.001]}{-0.646}$ & $\underset{[0.003]}{-0.295}$ \\[10pt]
$\phi_2$ & $\underset{[0.000]}{0.871}$ & $\underset{[0.000]}{0.037}$ & $\underset{[0.026]}{-0.902}$ & $\underset{[0.124]}{4.205}$ & $\underset{[0.001]}{0.803}$ & $\underset{[0.000]}{0.921}$ \\[10pt]
$\sigma_2^2$ & $\underset{[0.000]}{0.028}$ & $\underset{[0.000]}{0.006}$ & $\underset{[0.019]}{0.516}$ & $\underset{[0.066]}{3.548}$ & $\underset{[0.000]}{0.019}$ & $\underset{[0.000]}{0.038}$ \\[10pt]
$\rho_2$ & $\underset{[0.000]}{-0.950}$ & $\underset{[0.000]}{0.040}$ & $\underset{[0.020]}{1.254}$ & $\underset{[0.127]}{4.806}$ & $\underset{[0.000]}{-0.995}$ & $\underset{[0.001]}{-0.871}$ \\[10pt]
\\
$\log\hspace{1mm}p(y)$&$-7.72\times 10^{-3}$\\[10pt]
$\widehat{\ESS}_{\IS}/\widehat{\ESS}_{\IS^2} $ & 0.99\\
\hline\hline
\end{tabular}
\end{threeparttable}
\end{center}
\end{table}

\section{Conclusions}\label{Sec:conclusion}
This article proposes the IS$^2$ method for Bayesian inference when the likelihood is intractable but can be estimated unbiasedly.
Its advantages are that it gives accurate estimates of posterior moments and the marginal likelihood and estimates of the standard errors of these estimates.
 We note that if the model is estimated by IS$^2$, then the marginal likelihood
estimate and the standard error of the estimator are obtained automatically. However, even if the model is estimated by MCMC, IS$^2$ can be used effectively
to estimate the marginal likelihood and the corresponding standard error.
The article studies the convergence properties of the $\IS^2$ estimators. It
examines the effect of estimating the likelihood
on Bayesian inference
and provides practical guidelines on how to optimally select number of particles to estimate the  likelihood
in order to minimize the computational cost.
The applications illustrate that
the $\IS^2$ method can lead to fast and accurate posterior inference when optimally implemented,
and demonstrate  that the theory and methodology presented in this paper
are useful for practitioners who are working on models with an intractable likelihood.

\section{Supplementary material}  \label{S: supp material}
The paper has an online technical supplement that contains: (a) Some large sample results for panel data about the estimator of the likelihood and
static and dynamic estimators of the optimal variance of the log-likelihood; (b)~proofs of all propositions; (c) Discussion of an alternative importance sampling density
$g_{IS}(\theta)$ and the definition of equivalent sample size (ESS).

\begin{appendix}
\section{Proofs}
\begin{proof}[Proof of Theorem \ref{the:IS1}]
If $\Sup(\pi)\subseteq\Sup(g_\text{IS})$ then $\Sup(\pi_N)\subseteq\Sup(\wt g_\text{IS})$.
This, together with the existence and finiteness of $\E_\pi(\varphi)$ ensures that $\E_{\wt g_\IS}[\varphi(\t)\wt w(\t,z)]=p(y)\E_\pi(\varphi)$ and $\E_{\wt g_\IS}[\wt w(\t,z)]=p(y)$ exist and are finite.
Result (i) follows from the strong law of large numbers.
To prove (ii), write
\beqn
\wh\varphi_{\IS^2}-\E_\pi(\varphi) = \frac{\frac1M\sum_{i=1}^M\big(\varphi(\t_i)-\E_\pi(\varphi)\big)\wt w(\t_i,z_i)} {\frac1{M}\sum_{i=1}^{M}\wt w(\t_i,z_i)}.
\eeqn
Let $X_i=\big(\varphi(\t_i)-\E_\pi(\varphi)\big)\wt w(\t_i,z_i)$, $i=1,...,M$, $S_{M}=\frac{1}{M}\sum_{i=1}^{M} X_i$ and $Y_{M}=\frac1{M}\sum_{i=1}^{M}\wt w(\t_i,z_i)$.
The $X_i$ are independently and identically distributed with $\E_{\wt g_\IS}(X_i)=0$ and
\bean
\nu^2_N=\Var_{\wt g_\IS}(X_i)=\E_{\wt g_\IS}(X_i^2)&=&\int_{\wt\Theta}\big(\varphi(\t)-\E_\pi(\varphi)\big)^2\wt w(\t,z)p(y)\pi_N(\t,z)d\t dz\\
&=&p(y)\E_{\pi_N}\left\{\big(\varphi(\t)-\E_\pi(\varphi)\big)^2\wt w(\t,z)\right\}<\infty.
\eean
By the CLT for a sum of independently and identically distributed random variables
with a finite second moment, $\sqrt{M}S_{M}\stackrel{d}{\to}\N(0,\nu^2_N)$.
By the strong law of large numbers, $Y_{M}\stackrel{P}{\to}\E_{\wt g_\IS}(\wt w(\t,z))=p(y)$.
By Slutsky's theorem,
\beqn
\sqrt{M}\Big(\wh\varphi_{\IS^2}-\E_\pi(\varphi)\Big)=\frac{\sqrt{M}S_{M}}{Y_{M}}\stackrel{d}{\to}\N(0,\nu^2_N/p(y)^2).
\eeqn
The asymptotic variance is given by
\bean
\sigma^2_{\IS^2}(\varphi)=\frac{\nu^2_N}{p(y)^2}&=&\frac{1}{p(y)}\E_{\pi_N}\left\{\big(\varphi(\t)-\E_\pi(\varphi)\big)^2\wt w(\t,z)\right\}\\
&=&\E_\pi\left\{\big(\varphi(\t)-\E_\pi(\varphi)\big)^2\frac{\pi(\t)}{g_\IS(\t)}\E_{g_N(z|\t)}[\exp(2z)]\right\}.
\eean
To prove (iii), write
\bean
\wh{\sigma^2_{\IS^2}(\varphi)}&=&\frac{\frac{1}{M}\sum_{i=1}^{M}\big(\varphi(\t_i)-\wh\varphi_{\IS^2}\big)^2\wt w(\t_i,z_i)^2}{\left(\frac{1}{M}\sum_{i=1}^{M}\wt w(\t_i,z_i)\right)^2}\\
&\stackrel{a.s.}{\to}&\frac{\E_{\wt g_\IS}\left\{\big(\varphi(\t)-\E_\pi(\varphi)\big)^2\wt w(\t,z)^2\right\}}{\Big(\E_{\wt g_\IS}(\wt w(\t,z))\Big)^2}\\
&=&\sigma^2_{\IS^2}(\varphi).
\eean
\end{proof}

\begin{proof}[Proof of Theorem \ref{lem: robust}]
Proof of the unbiasedness of $\wh p_{\IS^2}(y)$ is straightforward.
The assumptions in Theorem \ref{the:IS1} ensure that $\wt w(\t_i)$'s are i.i.d with a finite second moment.
The results of (i) and (iii) then follow.
To prove (ii), we have
\bean
M\Var_{\wt g_\IS}(\wh p_{\IS^2}(y))=\Var_{\wt g_\IS}(\wt w(\t))&\leq&\E_{\wt g_\IS}[\wt w(\t)^2]\\
&=&p(y)^2\int\left(\int e^{2z}g_N(z|\t)dz\right)\left(\frac{\pi(\t)}{g_\IS(\t)}\right)^2g_\IS(\t)d\t<\infty.
\eean
\end{proof}

\begin{proof}[Proof of Theorem \ref{the:ISefficency}]
Under Assumption 2, $g_N(z|\t)=\N(-\s^2/2,\s^2)$ and
$\E_{g_N(z|\t)}[\exp(2z)]=\exp(\s^2)$.
From \eqref{e:ISvar_est_llh} and \eqref{e:ISasym_var},
\beqn
\sigma^2_{\IS^2}(\varphi)=\E_\pi\left\{\big(\varphi(\t)-\E_\pi(\varphi)\big)^2\frac{\pi(\t)}{g_\IS(\t)}\exp(\sigma^2)\right\}=\exp(\sigma^2)\sigma^2_\IS(\varphi).
\eeqn
\end{proof}
\end{appendix}
\renewcommand{\theequation}{S\arabic{equation}}
\renewcommand{\thesection}{S\arabic{section}}
\renewcommand{\theproposition}{S\arabic{proposition}}
\renewcommand{\theassumption}{S\arabic{assumption}}
\renewcommand{\thelemma}{S\arabic{lemma}}
\renewcommand{\thefigure}{S\arabic{figure}}
\renewcommand{\thelemma}{S\arabic{table}}
\setcounter{section}{0}

\section*{Supplementary material for \lq Importance Sampling Squared for Bayesian Inference and Model Choice with Estimated Likelihoods\rq}

\section{Prelminaries}
This supplement contains results that complement those in the main paper. References to the main paper are of the form section~2, equation~(2) and assumption~2, etc,
whereas for the supplementary material we use section~S2, equation (S2) and assumption S2, etc.

\section{Large sample results for panel data}\label{sec:IS_theory}
This section studies the large sample  (in $n$) properties of the
estimator of the likelihood
for  panel data which is described in section~\ref{Sec:latent variable models}, and obtains two results assuming that the number of particles
is $N = O(n)$. First, we justify Assumption \ref{ass: normal} by showing the asymptotic normality
of the error $z$ in the log-likelihood estimator.
Second, we obtain the convergence rates of two estimators of the optimal variance, $\sigma^2_{\rm opt}$,
of the error $z$ in the log likelihood estimator. We show  that the dynamic estimator of $\sigma^2_{\rm opt}$
that chooses the number of particles or samples $N$ depending on $\theta$ is much more efficient for large $n$ than the
static estimator of $\sigma^2_{\rm opt}$ that chooses a single number of particles $N$ for all $\theta$.
Using the notation in section~\ref{Sec:latent variable models}, the error in the log-likelihood estimator is
\begin{align}\label{eq: exp for z}
z  & =\sum_{i=1}^{n} \{\log\;\widehat{p}_{N,i}(y_{i}|\theta)-\log\;p_i(y_{i}|\theta)\}, \text{
where   } {\wh {p}}_{N,i}(y_{i}|\theta)= N^{-1} { \sum_{j=1}^{N}}
\omega_i(\a_{i}^{(j)},\theta),
\end{align}
 $\omega_{i}(\a_{i},\theta)=p(y_{i}|\a_{i},\theta)p(\a_{i}|\theta)/h_i(\a_{i}|y,\theta)$,
and $\a_{i}^{(j)}\stackrel{iid}{\sim}h_i(\a_{i}|y,\theta)$.

All equations, sections etc in this supplement are prefixed by S, e.g. equation (S1). Equations, sections etc that are not so prefixed refer to the main paper.
\subsection{Large sample properties of the estimator of the likelihood}\label{SS app: large sample properties}
Assumption~\ref{ass: ineq} is used to prove Proposition~\ref{prop:large sample z}.
\begin{assumption} \label{ass: ineq}
We assume that for all $\theta$ and $i=1, \dots, n$,
\begin{enumerate}
\item [(i)]
The proposal density $h_i(\a_{i}|y,\theta)$ is sufficiently heavy
tailed so that $\mathbb{E}_{h_i}\left \{  \omega_{i}(\a_{i},\theta
)^{4}\right \}  < K_1 $, where $K_1 > 1 $.
\item [(ii)]   $K_2^{-1} < p(y_i|\theta) \leq K_2 $, where  $K_2 > 1 $, so that $p(y_i|\theta)$ is bounded and also bounded away from zero.
\item [(iii)]  $\mathbb{V}_{h_i}\left \{  \omega_{i}(\a_{i},\theta
)^{2}\right \}  > K_3 > 0  $, so that the importance weights are not all equal, i.e., the proposal is not perfect.
\end{enumerate}
\end{assumption}

Define,
\begin{align}\label{eq: sigma_i}
\sigma_{i}^{2}(\theta) &  :=\mathbb{V}_{h_i} \bigg ( \frac{\omega_i(\alpha_i, \theta) }{ p_i(y_{i}|\theta) }\Big |\theta  \bigg )
= \int \frac{ \omega_{i}(\a_{i},\theta)^{2}}{ p_i(y_{i}|\theta)^{2}} h_i(\a_{i}|y,\theta) d\alpha_i-1.\\
\psi(n,\theta) &  := \frac1n \sum_{i=1}^n \sigma_i^2(\theta) , \quad  \wh { \psi}_N(n,\theta)  := \frac1n \sum_{i=1}^n \wh {\sigma}_i^2(\theta), \text{  where   }
\widehat{\sigma}_{i}^{2}(\theta):= \frac1N\sum_{j=1}^{N} \Big ( \frac{ \omega_{i}(\a_{i}^{(j)},\theta)^{2}}{   \widehat{p}_{N,i}(y_{i}|\theta)^2   } -1\Big )  ,  \label{eq_vhat}\\
\wh { \mathbb{E}}(z)& :=-\frac{n\wh {\psi}_N(\theta,n)}{2N} \text{  and  }
\wh { \mathbb{V}} \left(  z\right)  : =\frac{ n\wh {\psi}_N (\theta,n)}{N} . \label {eq: est mom}
\end{align}

Proposition~\ref{prop:large sample z}  gives the following asymptotic results for the error $z$. Its proof is in the section~\ref{app supp: proof}.
\begin{proposition} \label{prop:large sample z}  
Suppose Assumptions~\ref{ass: exp} and \ref{ass: ineq} hold. Then,
\begin{enumerate}
\item [(i)]
\begin{align}\label{eqn:z_moments}
\mathbb{E} (z|n,N, \theta) & =-n\psi(\theta,n)/(2N) +O(n/N^{2}) \quad \text{and} \quad
\mathbb{V}\left( z | n,N,\theta\right) =n\psi(\theta,n)/N+O(n/N^{2}).
\end{align}
\item [(ii)]
\begin{align}
\widehat{\psi}_{N}(\theta,n)  &  =
\psi(\theta,n)+O\left(  N^{-1}\right)  +O_{p}(n^{-\frac{1}{2}}N^{-\frac
{1}{2}}).
\end{align}
\item [(iii)]
\begin{align} \label{eq: est momts}
\wh { \mathbb{E}}(z)& = \mathbb{E}(z ) +  O(nN^{-2}) +  O_p(n^{1/2}N^{-3/2}) \text{  and  }
\wh { \mathbb{V}}  = \mathbb{V}(z^{} ) + O(nN^{-2}) +  O_p(n^{1/2}N^{-3/2}).
\end{align}
\item [(iv)]
\begin{align}
\mathbb{V}(z)^{-1/2}\{z -\mathbb{E}(z  )\} & \overset{d}{\longrightarrow}%
\mathcal{N}(0,1)\text{ and }\widehat{\mathbb{V}}(z)^{-1/2}\{z-\widehat
{\mathbb{E}}(z)\} \overset{d}{\longrightarrow}\mathcal{N}(0,1).
\label{eq:IS_CLT}%
\end{align}
\end{enumerate}
\end{proposition}

\begin{remark}
\begin{itemize}
\item [(1)] Proposition~\ref{prop:large sample z} justifies Assumption~\ref{ass: normal}, with $\gamma^2(\theta) = n \psi(n,\theta) $.
\item [(2)]
The convergence of $z$ to normality is fast, because $N = O(n)$ and from Part (vii) of Lemma~\ref{lemma: prop eps} and
\eqref{eq: exp for z}, $z$ is a sum of $n$ terms, each of which converges to normality.
\end{itemize}
\end{remark}

\section{ Estimating the optimal number of particles statically and dynamically} \label{SSS: static dynamic number of particles}
In practice, we adopt two strategies for selecting $N$. The first strategy, which we call a static strategy, is to select the number of particles $N$ for a given value
${\wh \theta}$ of $\theta$ so that
that $N = N_{\sigma^2}({\wh \theta}) $ with $\frac{ \gamma^2({\wh \theta} ) }{ N_{\sigma^2}({\wh \theta})} = \sigma^2_{\rm opt} $. Usually,
${\wh \theta}$ is an estimate of a central value of $\theta$ such as the posterior mean. The static strategy is quite general and has the advantage that it does not require tuning the number of particles for each $\theta$. The dynamic strategy
selects the number of particles to try and target $\frac{ \gamma^2({ \theta} ) }{ N_{\sigma^2}({\theta})} = \sigma^2_{\rm opt} $ for each $\theta$.
Under suitable regularity conditions, Proposition~\ref{proposition:var_consist} shows  that asymptotically in the sample size $n$
the dynamic strategy is superior to the static strategy because the estimated variance $\wh\Var_{N,\t}(z)$ of $z$ under the dynamic strategy converges to its target value much in order $O_P(n^{-1})$, which is much faster than the static strategy which converges in order  $O_P(n^{-\frac12})$.
The proof is in section~\ref{app supp: proof}.

\begin{proposition}\label{proposition:var_consist}  
Suppose the following five conditions hold.
\begin{enumerate}
\item[(i)] For any in $\theta \in \Theta$,
\begin{equation}
\mathbb{V}(z|\theta,n,N)=n\psi(\theta,n)/N+O(n/N^{2}),
\label{eq: sigmaZsq asymptot}%
\end{equation}
where $\psi(\theta,n)$ is $O(1)$ in $n$ uniformly in $\theta \in \Theta$.
\item[(ii)] A consistent estimator $\widehat{\psi}_{N}(\theta,n)$ (in $N$) of
$\psi(\theta,n)$ is available such that, uniformly in $\theta$,
\beqn
\widehat{\psi}_{N}(\theta,n)=\psi(\theta,n)+O\left(  N^{-1}\right)
+O_{p}(n^{-\frac{1}{2}}N^{-\frac{1}{2}}).
\eeqn
\item[(iii)] The posterior $\pi(\theta)$ satisfies
$\lim_{n\rightarrow \infty}\mathbb{E}_{\pi}(\theta)=\overline{\theta}$ and
$\lim_{n\rightarrow \infty}n\times \mathbb{V}_{\pi}(\theta)=\Sigma$ .

\item[(iv)] We have a consistent estimator (in $n$) $\widehat{\theta}_{n}$ of
$\overline{\theta}$.

\item[(v)] $\psi(\theta,n)$ is a continuously differentiable function of
$\theta \in \Theta$.
\end{enumerate}
Then,
\begin{enumerate}
\item[A.] Under Conditions (i) to (v), suppose that an initial starting value
$N_{S}=O\left(  n\right)  $ is used and the number of particles $N_{\widehat
{\theta}_{n}}$ is chosen statically as $N_{\widehat{\theta}_{n}}%
:=n\widehat{\psi}_{N_{S}}(\widehat{\theta}_{n},n)/\kappa^{2}$, where
$\kappa^{2}$ is a prespecified target variance. Then,  $\mathbb{V}(z|\theta,n,N_{{\wh \theta}_n})=\kappa^{2}+O_{p}(n^{-1/2})$.
\item[B.] If conditions  (i) and (ii) hold  and  $N_{\theta}$ is chosen
dynamically as $N_{\theta}:=n\widehat{\psi}%
_{N_{\widehat{\theta}_{n}}}(\theta,n)/\kappa^{2}$, where $N_{\widehat{\theta
}_{n}}=O\left(  n\right)  +O_{p}\left(  1\right)  $, e.g. as in result A, then
$\mathbb{V}(z|\theta,n,N_{\theta})=\kappa^{2}+O_{p}(n^{-1}).$
\end{enumerate}
\end{proposition}
We note that Conditions (i) and (ii) are justified for the IS estimator by Parts~(i) and (ii) of Proposition~\ref{prop:large sample z}.
The dynamic estimator is preferred because it converges faster and does not
require any assumption on the behavior of the posterior $\pi(\theta)$.
The efficiency of the dynamic estimator is shown empirically in section \ref{sec:likeval1} of the article.
For part (iv),  we note that any consistent estimator, such
as a method of moments estimator, may be employed. However,
the dynamic estimator is more expensive than the static estimator as it also requires  computing $\widehat{\psi}_{N_{\widehat\t_n}}(\t,n)$.

\section{Proofs of the Propositions} \label{app supp: proof}
\begin{proof}[Proof of Proposition \ref{proposition:marg_lik}]
Denote $a=\tau_0$, $b=\tau_1\bar\gamma^2$, and $x=\sigma^{2}$. Then
\beqn
f(x) = \CT_{\wh p_{\IS^2}(y)}(\s^2) = \big(a+\frac{b}{x}\big)\big((v+1)e^x-1\big).
\eeqn
We write $f(x) = f_1(x)+bf_2(x)$ with $f_1(x)=a((v+1)e^x-1)$ and $f_2(x)=((v+1)e^x-1)/x$.
As $f_1(x)$ is convex, it is sufficient to show that $f_2(x)$ is convex.
To do so, we will prove that $f_2''(x)>0$ for $x>0$, because a differentiable function is convex
if and only if its second derivative is positive \citep[see, e.g.,][Chapter 3]{Bazaraa:2006}.

After some algebra
\beqn
f_2''(x)=v\frac{e^{x}}{x}\left[  \left(  1-\frac{1}{x}\right)
^{2}+\frac{1}{x^{2}}\right]  +\frac{e^{x}}{x}\left[  \left(  1-\frac{1}%
{x}\right)  ^{2}+\frac{1}{x^{2}}\left(  1-2e^{-x}\right)  \right].
\eeqn
The first term is clearly positive. We consider the second
term in the square brackets
\[
\left(  1-\frac{1}{x}\right)  ^{2}+\frac{1}{x^{2}}\left(  1-2e^{-x}\right)
=1-\frac{2}{x}+\frac{2(1-e^{-x})}{x^{2}}.
\]
Note that $1-e^{-x}>x-\frac{x^{2}}{2}$ for $x>0$ and so
\[
1-\frac{2}{x}+\frac{2}{x^{2}}(1-e^{-x})>1-\frac{2}{x}+\frac{2}{x^{2}}\left(
x-\frac{x^{2}}{2}\right)  =0.
\]
This establishes that $f_2^{\prime\prime}(x)>0$ for all $x>0$.

To prove (ii), for any fixed $v$, let $x_{\min}(v)$ be the minimizer of $f(x)$.
By writing
\beqn
f(x)=(v+1)\CT(x)\times \left(1-\frac{a+b/x}{(v+1)\CT^*(x)}\right),
\eeqn
with $\CT(x)$ defined by \eqref{eq: CT def},
we can see that $f(x)$ is driven by the factor $(v+1)\CT(x)$ as $v\to\infty$.
Hence $x_{\min}(v)$ tends to the $\sigma^2_\text{opt}$ in \eqref{eq:sigma_opt}
that minimizes $\CT(x)$, as $v\to\infty$.

Because $f'(x_{\min}(v))=0$ for any $v>0$,
\beq\label{eq: xmin v}
e^{x_{\min}(v)}(ax_{\min}^2(v)+bx_{\min}(v)-b)=-\frac{b}{v+1},\;\;\text{for}\; v>0.
\eeq
By taking the first derivative of both sides of \eqref{eq: xmin v},
we have
\beqn
x_{\min}'(v)=\frac{b}{(v+1)^2}\frac{e^{-x_{\min}(v)}}{ax_{\min}^2(v)+(b+2a)x_{\min}(v)}>0,\;\;\text{for any}\; v>0.
\eeqn
This follows that $x_{\min}(v)$ is an increasing function of $v$.
Furthermore, $x_{\min}'(v)\to0$ as $v\to\infty$, which establishes (iii).
\end{proof}
To obtain Proposition~\ref{prop:large sample z} we first obtain some preliminary results.
Define,
\begin{align}
\zeta_{ij} (\theta) & := \omega_i (\alpha_i^{(j)} , \theta)/p(y_i|\theta) -1  \label{eq: def zeta}\\
\varepsilon_{N,i} & :=  \frac{\sqrt{N}}{ \sigma_i(\theta)}  \left(  \frac{  \widehat{p}_{N,i}(y_{i}|\theta)}{ p_i(y_{i}%
|\theta)} -1\right)   = \frac{1}{\sqrt N \sigma_i(\theta) } \sum_{j=1}^N \zeta_{ij} (\theta)\label{eq: def vareps}
\end{align}
so that $\varepsilon_{N,i}$ is a sum of the i.i.d. random variables.

\begin{lemma} \label{lemma: prop eps}
Suppose Assumptions~\ref{ass: exp} and \ref{ass: ineq} hold. Then, for all $i = 1, \dots, n$, and $\theta$,
\begin{enumerate}
\item [(i)]
$\mathbb{E}(\zeta_{ij} (\theta)) =  0 $ and $\mathbb{V}(\zeta_{ij} (\theta)) =  \sigma_i^2(\theta) $.
\item [(ii)]
$ \mathbb{E}_{h_i} (|\zeta_{ij} (\theta)|^k ) \leq K_4 $ for all $j = 1, \dots, N,$ and $ k=1, \dots, 4$, with $K_4> 0 $.
\item [(iii)]
\begin{align}
N\times \mathbb{V}_{h_i} \left \{  \widehat{p}_{N,i}(y_{i}|\theta)/p(y_{i}|\theta)\right \} &  = \sigma_{i}^{2}(\theta)
 \label{eq:moments_ISt}
 \end{align}
\item [(iv)]
$1/K_5  < \sigma_i(\theta)^2 < K_5$, where $K_5 > 1$ is a constant.
\item [(v)]
$\mathbb{E}( \varepsilon_{N,i}| \theta)  = 0 $, $\mathbb{V}( \varepsilon_{N,i}|\theta) = 1$,
$\mathbb{E}( \varepsilon_{N,i}^3| \theta)  = O(N^{-1/2}) $,  and $\mathbb{E}(| \varepsilon_{N,i}|^k| \theta) < K_5$, for $ k = 1, \dots, 4$,
where $K_5> 0 $.
 \item [(vi)]
\begin{align}
\frac{\mathbb{E}( \varepsilon_{N,i}^3|\theta) }  {\big ( \mathbb{E}( \varepsilon_{N,i}^2|\theta)\big )^{3/2}  }
& \rightarrow 0  \text{  as   }  N \rightarrow \infty. \label{eq: clt condn 1}
\end{align}
 \item [(vii)]
$\varepsilon_{N,i}\overset{d}{\rightarrow}\mathcal{N}(0,1)$ as $N\rightarrow \infty$.
\end{enumerate}
\begin{proof}
Parts (i) and (ii)  follow from Assumption~\ref{ass: ineq}, Assumption~\ref{ass: exp} and \eqref{eq: sigma_i}.
Part (iii) follows from the definition \eqref{eq: def vareps}.
Part (iv) follows from Assumption~\ref{ass: ineq}. To obtain part (v), we note that $\mathbb{E}( \varepsilon_{N,i}| \theta)  = 0 $,
$\mathbb{V}( \varepsilon_{N,i}|\theta) = 1$, follow from the definition of
$\varepsilon_{N,i}$ and Part (i). It is straightforward to show that $\mathbb{E}( \varepsilon_{N,i}^3|\theta) = N\mathbb{E}(\zeta_{ij} (\theta)^3)/(N^{3/2} \sigma_i(\theta)^3)  = O(N^{-1/2} )$ by Part (i).
Part (v) follows from the definition of
$\varepsilon_{N,i}$ Part (i), and Assumption~~\ref{ass: ineq}.
The denominator on the left side of \eqref{eq: clt condn 1} is
$O(1)$. The numerator is $O(N^{-1/2})$ by part (v). Part (vi) follows. Part (vii) holds  because
$\varepsilon_{N,i}$ is a sum of i.i.d. random variables, Part (v) and Theorem 5, p. 194, of \cite{stirzaker2001probability}.
\end{proof}
\end{lemma}
We write the error $z$ in \eqref{eq: exp for z} in terms of the
$\varepsilon_{N,i}$ and expand as a third order Taylor series approximation with a remainder term,
\begin{align}
z & = \sum_{i=1}^{n}
\log \left(  1+\sigma_{i}(\theta)\varepsilon_{N,i}/\sqrt{N}\right)
\notag\\
&= \sum_{i=1}^{n}\left(  \frac{\sigma_{i}(\theta)\varepsilon_{N,i}%
}{\sqrt{N}}
-\frac{\sigma_{i}^{2}(\theta)\varepsilon_{N,i}^{2}}{2N}%
 +\frac{\sigma_{i}^{3}(\theta)\varepsilon_{N,i}^{3}}{3N\sqrt{N}} +
R_{N,i}(\theta)
  \right ) . \label{eq:z_decomp}
\end{align}
The remainder term is $R_{N,i}(\theta) \leq
\sigma_{i}^{4}(\theta)\varepsilon_{N,i}^{4}/ 4N^{2}$.

\begin{lemma} \label{lemma: sigmahat_i}
We can express $\wh {\sigma}_i^2 (\theta) = \sigma_i^2(\theta) + c_i(\theta)/N + \xi_i(\theta)/\sqrt{N} $,
where $\mid c_i(\theta) \mid < K_5$ and $\mathbb{V}_{h_i}(\xi_i(\theta)) < K_5$, with $K_5> 0 $,  for all $i=1, \dots, n$ and $\theta$.
\begin{proof}
\begin{align*}
\Big ( \frac{{\wh p_{N,i}}(y_i|\theta)}{p_i(y_i|\theta)}   \Big )^2   = \Big ( \frac1N \sum_{j=1}^N \big ( 1 + \zeta_{ij}(\theta) \big )\Big)^2 & =
 1 + \frac{1}{N^2} \Big ( \sum_{j=1}^N \zeta_{ij}(\theta)  \Big ) ^2 + \frac2N \sum_{j=1}^N \zeta_{ij}(\theta)\\
\text{ so that} \quad
\mathbb{E}_{h_i}\Big ( \frac{{\wh p_{N,i}}(y_i|\theta)}{p_i(y_i|\theta)}   \Big )^2  & = 1 + \frac{\sigma_i^2(\theta)}{N}\\
\text{and therefore we can write} \quad \Big ( \frac{{\wh p_{N,i}}(y_i|\theta)}{p_i(y_i|\theta)}   \Big )^2  & =   1 + \frac{\sigma_i^2(\theta)}{N}  + \frac{\xi_i^\ast(\theta) }{ \sqrt N}
\end{align*}
where $\mathbb{V}_{h_i} \big ( \xi^\ast (\theta)\big ) < K_6$, with $K_6 > 0 $ for all $\theta$ and $i = 1, \dots, N$.
Define,
\begin{align*}
{\wt \sigma}_i^2 (\theta) &  := \frac1N  \sum_{j=1}^N \bigg ( \frac{ w_i (\alpha_i^{(j),\theta}} {p_i(y_i(\theta) } \bigg )^2 -1  =
\frac1N \sum_{j=1}^N \zeta_{ij}(\theta)  ^2 + \frac2N \sum_{j=1}^N \zeta_{ij}(\theta) .
\end{align*}
We can then write ${\wt \sigma}_i^2 (\theta)  = \sigma_i^2(\theta)  + N^{-1/2} {\wt \xi}(\theta)$,
where $\mathbb{V}_{h_i} \big ( {\wt \xi} (\theta)\big ) < K_7$, with $K_7 > 0 $ for all $\theta$ and $i = 1, \dots, n$.

With some algebra, we can write
\begin{align*}
{\wh \sigma}_i^2(\theta) & = \Big ( \frac{{\wh p_{N,i}}(y_i|\theta)}{p_i(y_i|\theta)}   \Big )^{-2} \Bigg (   {\wt \sigma}_i^2 (\theta) + \bigg  ( \Big (1 -  \frac{{\wh p_{N,i}}(y_i|\theta)}{p_i(y_i|\theta)}   \Big )^{2}                     \Bigg )
\end{align*}
and the result follows.
\end{proof}\end{lemma}
\begin{proof}[Proof of Proposition~\ref{prop:large sample z} ]
Part (i): The expression for $\mathbb{E}(z|n,N,\theta)  $  follows from \eqref{eq:z_decomp}. To obtain the expression for $\mathbb{V}(z|n,N,\theta)  $,
we write $z$ as a first order Taylor series expansion plus remainder,
similarly to \eqref{eq:z_decomp},
\begin{align} \label{eq: exp z1}
z & = z_1 - \sum_{i=1}^n S_{N,i} (\theta), \quad \text{where} \quad z_1 = N^{-1/2} \sum_{i=1}^{n} \sigma_{i}(\theta)\varepsilon_{N,i}%
 \quad \text{and} \quad
S_{N,i}(\theta) \leq \sigma_i^2(\theta) \varepsilon_{N,i}^2 /(2N).
\end{align}
Part (ii) follows from Lemma~\ref{lemma: sigmahat_i}. Part (iii) follows from Part (ii). To obtain Part (iv), it is sufficient to prove the central limit theorem for
$z_1$, ,which is a sum of independent random variables. Now, by Parts (iv) and (v) of Lemma \ref{lemma: prop eps},
\begin{align*}
\frac{\sum_{i=1}^n \sigma_i^3(\theta) \mathbb{E}(\varepsilon_{N,i}^3)/N^{3/2}} { \Big ( \sum_{i=1}^n \sigma_i^2(\theta) \mathbb{E}(\varepsilon_{N,i}^3)/N^{1/2}\Big )^{3/2} } & = O(N^{-1/2}) \frac{\sum_{i=1}^n \sigma_i^3(\theta)} { \Big ( \sum_{i=1}^n \sigma_i^2(\theta)\Big )^{3/2} } \longrightarrow 0 \quad \text{as} \quad n  \longrightarrow \infty.
\end{align*}
The central limit theorem now follows from Theorem 5, p. 194, of \cite{stirzaker2001probability}
\end{proof}

We note that as $N$ is $O(n)$, the central limit
operates very quickly in $n$ upon the expression for $z_1$ in \eqref{eq: exp z1} as the
summation is over an increasing number of terms, each of which satisfies a
central limit.

\begin{proof}
[Proof of Proposition \ref{proposition:var_consist}]Assumption (ii) yields%
\[
N_{\widehat{\theta}_{n}}=n\widehat{\psi}_{N_{S}}(\widehat{\theta}%
_{n},n)/\kappa^{2}=n\psi(\widehat{\theta}_{n},n)/\kappa^{2}+O\left(
n/N_{S}\right)  +O_{p}(n^{\frac{1}{2}}N_{S}^{-\frac{1}{2}}).
\]
As $N_{S}$ is chosen proportional to $n$ then consequently, $N_{\widehat
{\theta}_{n}}$ is of order $n$. Using assumption (i), the achieved variance at
any ordinate $\theta$ with $N_{\widehat{\theta}_{n}}$ particles is%
\begin{equation}
\mathbb{V}(z|\theta,n,N_{\widehat{\theta}_{n}})=n\psi(\theta,n)/N_{\widehat
{\theta}_{n}}+O(n/N_{\widehat{\theta}_{n}}^{2})=\kappa^{2}\psi(\theta
,n)/\psi(\widehat{\theta}_{n},n)+O_{p}(n^{-1}). \label{eq:var_ratioeq}%
\end{equation}
Using Assumptions~(iii) to (v), we obtain from a Taylor expansion of
$\psi(\theta,n)$ around $\widehat{\theta}_{n}$ in (\ref{eq:var_ratioeq}) that
\[
\mathbb{V}(z|\theta,n,N_{\widehat{\theta}_{n}})=\kappa^{2}+O_{p}%
(n^{-\frac{1}{2}}).
\]
This establishes part A. For part B, we may write
\[
N_{\theta}=n\widehat{\psi}_{N_{\widehat{\theta}_{n}}}(\theta,n)/\kappa
^{2}=n\psi(\theta,n)/\kappa^{2}+O_{p}\left(  1\right)  ,
\]
since $N_{\widehat{\theta}_{n}}$ is of order $n$. Hence, the achieved variance
at a given $\theta$ and $n$ is given by
\[
\mathbb{V}(z|\theta,n,N_{\theta})=n\psi(\theta,n)/N_{\theta}+O(n/N_{\theta
}^{2})=\kappa^{2}+O_{p}(n^{-1}).
\]

\end{proof}

\section{Discussion of complementary methodology} \label{S: extra methods}
\subsection{An annealed importance sampling approach for constructing the importance density $g_{\IS}(\theta)$ }\label{SS_supp:proposal}
The second approach for constructing the density $g_{IS}$
is based on the annealed importance sampling for models with latent variables (AISEL) method
proposed by \cite{duan:fulop:2015} for state space models and independently by \cite{tran2014_aisel} for general models with intractable likelihoods,
which is based on the annealed importance sampling (AIS) approach of \cite{Neal:2001}, but where the likelihood is now estimated.
We advocate taking a small number of steps at a high temperature to obtain $\theta$ particles to train the $g_{\rm IS}(\theta)$ proposal.
The overall AISEL proposal is far more expensive than IS$^2$, but taking a few steps of AISEL at high temperatures requires relatively few particles and so is relatively cheap. The AISEL method produces a sequence of weighted samples
of $\theta$'s which are first drawn from an easily-generated distribution,
and then moved towards the target distribution through Markov kernels.
A very good proposal density $g_\IS(\t)$ that is heavier tailed than the target density can then be obtained
by fitting a mixture of $t$ densities to the resulting empirical density for  $\theta$.

\subsection{Effective sample size} \label{S: ESS}
The efficiency of the proposal density $g_\IS(\t)$ in the standard IS procedure \eqref{e:standardISestimator}
is often measured by the effective sample size defined as \citep[p.35]{Liu:2001}
\beq
\text{ESS}_\IS: = \frac{M}{1+\Var_{g_\IS}(w(\t)/\E_{g_\IS}[w(\t)])}
=\frac{M\big(\E_{g_\IS}[w(\t)]\big)^2}{\E_{g_\IS}[w(\t)^2]}.
\eeq
A natural question is how to measure the efficiency
of $g_\IS(\t)$ in the $\IS^2$ procedure.
In the $\IS^2$ context, the proposal density is $\wt g_\IS(\t,z)=g_\IS(\t)g_N(z|\t)$,
with the weight $\wt w(\t,z)=w(\t)e^z$ given in \eqref{eq:weight_IS2}.
Hence, the ESS for $\IS^2$ is $\text{ESS}_{\IS^2} = {M\big(\E_{\wt g_\IS}[\wt w(\t,z)]\big)^2}/{\E_{\wt g_\IS}[\wt w(\t,z)^2]}$,
which is estimated by $\wh{\text{ESS}}_{\IS^2} = {(\sum_{i=1}^M\wt w(\t_i))^2}\big/{\sum_{i=1}^M\wt w(\t_i)^2}$.
Under Assumption 2, $\E_{\wt g_\IS}[\wt w(\t,z)] = \E_{g_\IS}[w(\t)]$
and $\E_{\wt g_\IS}[\wt w(\t,z)^2] = e^{\s^2}\E_{g_\IS}[w(\t)^2]$.
Hence,
\beq\label{eq:ESS}
\text{ESS}_\IS= \exp({\s^2})\times\text{ESS}_{\IS^2}.
\eeq
If the number of particles $N(\t)$ is tuned to target a constant variance $\s^2$ of $z$,
then \eqref{eq:ESS} enables us
to estimate $\text{ESS}_\IS$ as if the likelihood was given.
This estimate is a useful  measure of the efficiency of the proposal density $g_\IS(\t)$ in the $\IS^2$ context.

\section{Sensitivity of the computing time to $\Var_{g_\IS}[W_i]$
in estimating the marginal likelihood}\label{SS: optimal N marginal likel}
This table is referred to in section~\ref{SS: optimal N marg likel} and shows that the ratio
$\CT_{\wh p_{\IS^2}(y)}(\s^2_\text{opt})/\CT_{\wh p_{\IS^2}(y)}(\s^2_{\min})$ is insensitive to $v$ and so justifies
using $\s^2_\text{opt}$ instead of $\s^2_\text{min}$  when estimating the marginal likelihood.
\begin{table}[h]
\centering%
\caption{Sensitivity of the computing time to $\Var_{g_\IS}[W_i]$
in estimating the marginal likelihood.}
\label{tab:marg_lik}%
\begin{tabular}
[c]{c|c|c}\hline\hline
$v=\Var_{g_\IS}[W_i]$ & $\sigma_{\min}^2(v)$ & $\CT_{\wh p_{\IS^2}(y)}(\s^2_\text{opt})/\CT_{\wh p_{\IS^2}(y)}(\s^2_{\min})$\\
\hline
1	&0.12	&1.0199\\
5	&0.16	&1.0012\\
10	&0.16	&1.0003\\
100	&0.17	&1.0000\\    	
$\infty$&0.17	&1\\
\hline\hline
\end{tabular}
\end{table}
\section{Empirical performance of IS$^2$ for panel data.} \label{S: empirical perfor of IS$^2$ for panel data}
\section{Empirical comparison of the static and dynamic estimators of the optimal number of particles}\label{sec:likeval1}
This section investigates the empirical performance of the static and dynamic estimators of the optimal number of particles defined
in section~\ref{SSS: static dynamic number of particles}, for the panel data models discussed in section~\ref{sec:gmnlmodel} of the main paper.

We generated 1,000 draws for the parameter vector $\theta$ from $g_\IS(\t)$. For each parameter combination, we obtained 100 independent likelihood estimates using the same fixed number of particles for all individuals in the panel and by targeting a log-likelihood variance of 1 and 0.5 for the MIXL model using  the static approach (fixed $N$)
to determining $N$ and targeting a variance of 0.5 for each $\theta$ using the dynamic approach. Similarly, for the GMNL model, the static approach with fixed $N = 8000$ and $N = 16000$ particles targets a variance of 2.0 and 1.0 respectively, while the dynamic approach targets a variance of 1.0 for each $\theta$.

Table \ref{tab:lik1} and
Figure~\ref{fig:lik1} summarize the results and show that dynamically  selecting  $N$ by targeting a
  constant $\sigma^2$ is effective, leading, on average, to approximately the correct variance for the
  log-likelihood estimates without substantial deviations from the target precision across different parameter combinations.
  In comparison, the static estimator of $N$ that uses a fixed number of particles leads to greater variability in precision across different draws of $\theta$, especially for the GMNL model. The results also show that the log-likelihood estimates are consistent with the normality assumption as required by the theory. Figure \ref{fig:optimal} shows the average number of particles required for obtaining the target variance $\Var(\log\hspace{1mm} \widehat{p}(y_i|\theta))\approx \sigma^2/I$ for each individual. The figure shows a substantial variation in the average optimal number of particles across individuals, in particular for the GMNL model.

Table \ref{tab:lik1} also provides estimates of the required quantities for calculating the optimal $\sigma^2$ in \eqref{eq:sigma_opt}.
For the MIXL model, since the average variance when $N=24$ is 1.068, we estimate $\ov{\gamma^2}=1.068\times24=25.63$.
As the likelihood evaluation time using $N$ particles is determined as $\tau_0+\tau_1N$,
from the computing times at $N=24$ and $N=48$,
we obtain $\tau_0=0.067$ and $\tau_1=8.97\times10^{-5}$ seconds.
Using \eqref{eq:sigma_opt}, we conclude that $\sigma^2_\text{opt}$ for the MIXL model is approximately 0.17.
The optimal variance $\sigma^2_\text{opt}=1$ for the GMNL case because $\tau_0=0$ as there is no overhead cost in obtaining the state proposal.

\begin{table}[!h]
\caption{{The table shows the average variance, skewness and kurtosis of log-likelihood estimates over 1,000 draws of $\t$.
The table shows the average sample variance, the standard deviation of the sample variance, the skewness and kurtosis of the  estimates
of the log-likelihood across the different draws for $\theta$, as well as the proportion of parameter combinations for which the Jarque-Bera test rejects the null hypothesis of the normality of the log-likelihood estimates at the $5\%$ level.
For the MIXL model, the columns labeled $N = 24$ and $N = 48$ correspond to the static estimator of $N$ that targets a variance of 1 and $0.5$ respectively. The column headed $\sigma^2 \approx 0.5$ correspond to the dynamic estimator of $N$ which targets a variance of $0.5$.
A similar explanation holds of the GMNL estimator and the columns headed $N = 8000$ (static estimator with a target variance of 2.0), $ N= 16000$ (static estimator with a target variance of 1.0) and $\sigma \approx 1$ (dynamic estimator with a target of variance of 1.0.
The last row of the table reports computing times based on
a computer equipped with an Intel Xeon 3.40 GHz processor with four cores.
}
}\label{tab:lik1}
\begin{center}
\begin{tabular}{lccclccc}
\hline\hline
 & \multicolumn{ 3}{c}{MIXL (EIS)} &  & \multicolumn{ 3}{c}{GMNL (natural sampler) } \\
 & $N=24$ & $N=48$ & $\sigma^2\approx 0.5$ &  & $N=8,000$ & $N=16,000$ & $\sigma^2\approx 1$ \\
  \cline{2-4} \cline{6-8}
Variance & 1.068 & 0.545 & 0.482 &  & 1.986 & 1.029 & 1.023 \\
SD of variance & 0.159 & 0.079 & 0.055 &  & 0.500 & 0.295 & 0.162 \\
Skewness & 0.035 & 0.051 & 0.024 &  & 0.002 & 0.002 & -0.043 \\
Kurtosis & 3.017 & 3.035 & 2.991 &  & 2.929 & 2.928 & 2.953 \\
JB rejections (5\%) & 0.070 & 0.090 & 0.060 &  & 0.046 & 0.037 & 0.051 \\
Time (sec) & 0.069 & 0.071 & 0.072 &  & 0.916 & 1.860 & 2.084 \\
\hline\hline
\end{tabular}
\end{center}
\end{table}

\begin{figure}[h]
    \begin{center}
        \subfigure[$N=16,000$]{%
           \includegraphics[scale=0.35]{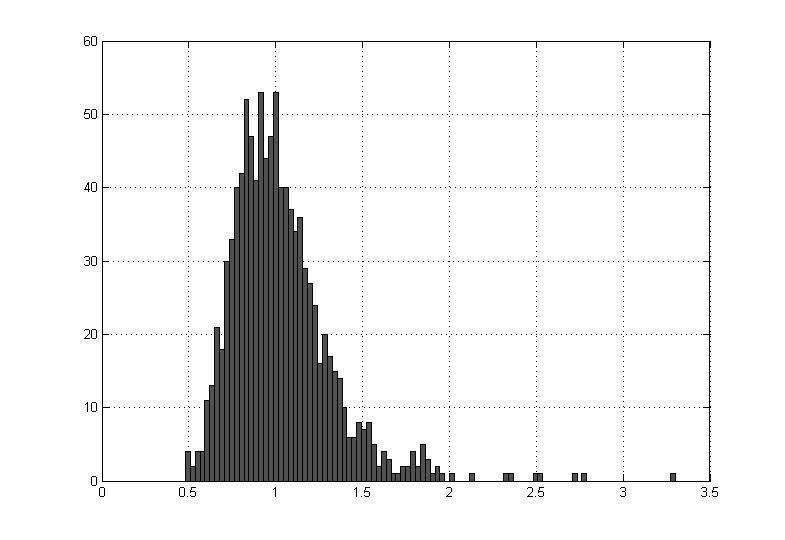}
        }%
        \subfigure[Targeting $\sigma^2=1$ for each draw of $\theta$]{%
\           \includegraphics[scale=0.35]{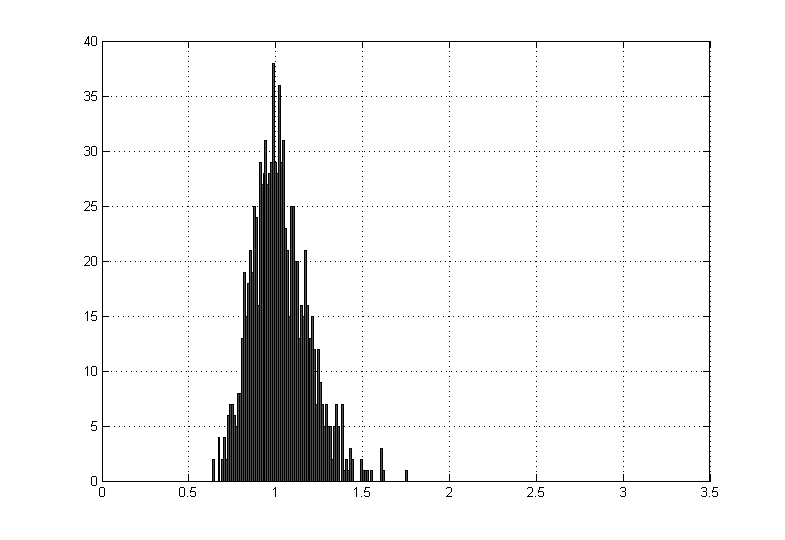}
        }%
    \end{center}
    \caption{The setup is the same as for Table \ref{tab:lik1} but refers only to the GNML model.
     Panel (a)
 displays the histogram of the sample variances  when a static estimator of $N$ with $N=16,000$ and targeting
 $\sigma^2 = 1 $ is used across the $1,000$ draws for $\theta$. Panel (b) displays the histogram of the sample variances  when the dynamic estimator of $N$ targets a variance of 1.0 for each $\theta$. }\label{fig:lik1}
\end{figure}

\begin{figure}
 \begin{center}
        \subfigure[MIXL]{%
           \includegraphics[scale=0.35]{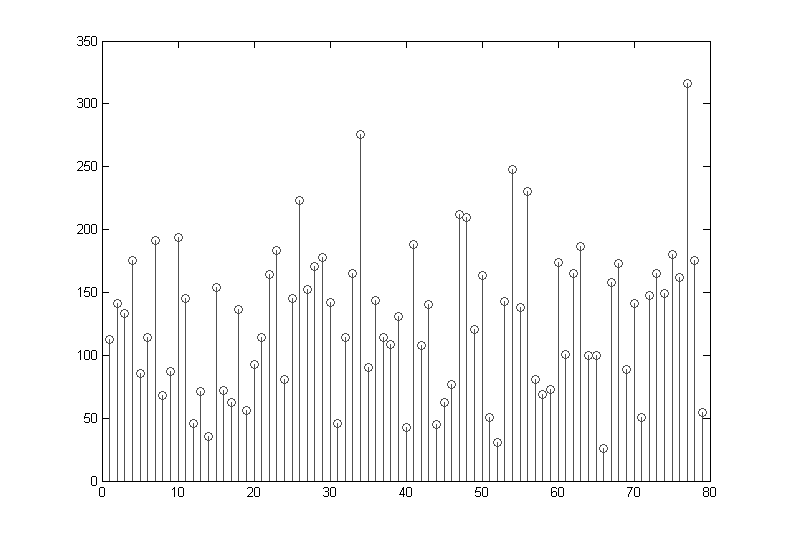}
       }%
        \subfigure[GMNL]{%
           \includegraphics[scale=0.35]{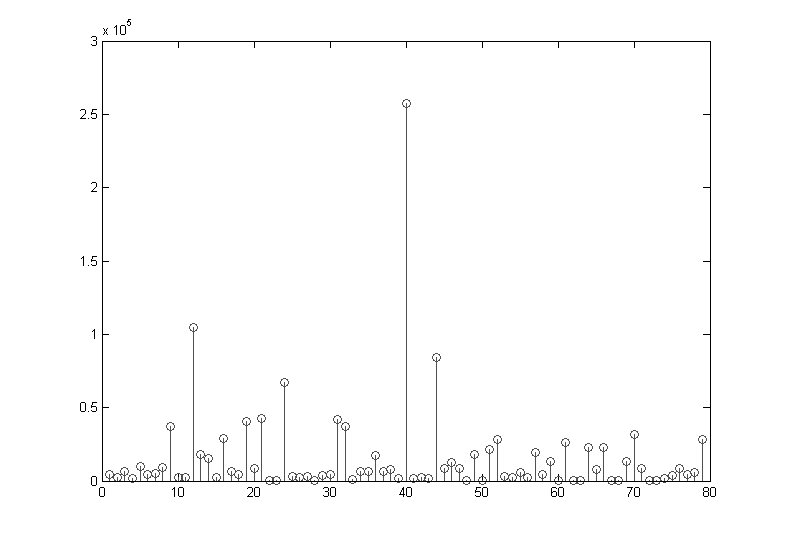}
        }%
    \end{center}
    \caption{The setup is the same as for Table \ref{tab:lik1}. The average number of particles for each individual when targeting $\sigma^2_\text{opt}$ using the dynamic estimator of $N$.}\label{fig:optimal}
\end{figure}

\subsection{Some further comparison of theoretical and empirical of performance of IS$^2$ for panel data}\label{SS: theoretical and empr compar panel data}
To compare different implementations of IS$^2$, we generate $M=50,000$ draws from the parameter proposal $g_\IS(\theta)$ and run the IS$^2$  by targeting different variances of the log-likelihood estimates. The target variances are $\sigma^2=0.05$, 0.2, 0.75, 1 for the MIXL model and $\sigma^2=0.5$, 1, 1.5, 2 for the GMNL model.  We estimate the Monte Carlo (MC) variances of the posterior means under each implementation by the bootstrap method. Tables \ref{tab:mixlexample1} and \ref{tab:example1} show the relative estimated MC variances, the total actual computing time in minutes, and the relative estimated time normalized variance \eqref{eq:TNV} for the MIXL and GMNL models respectively. We also report for reference the theoretical relative MC variances and TNV (see  \eqref{eq:Var_varphi} and \eqref{eq:TNV}).
\beqn
\frac{\Var(\wh\varphi_{\IS^2})|_{\s^2=\s_1^2}}{\Var(\wh\varphi_{\IS^2})|_{\s^2=\s_2^2}}=\frac{\exp(\sigma^2_1)}{\exp(\sigma^2_2)}\quad
\text{and}\quad
\frac{\text{TNV}(M,\s_1^2)}{\text{TNV}(M,\s_2^2)}=\frac{\exp(\sigma^2_1)(\tau_0+\tau_1\bar\gamma^2/\sigma^2_1)}{\exp(\sigma^2_2)(\tau_0+\tau_1\bar\gamma^2/\sigma^2_2)}.
\eeqn

The results show that the relative estimated MC variances are on average close to their theoretical values, even though there is a large degree of variability in the estimates across the parameters (especially for higher values of $\sigma^2$). The estimated relative TNV are approximately consistent with their theoretical values for the optimal precision of the log-likelihood estimates. Finally,  $\widehat{\ESS}_{\IS^2}/\widehat{\ESS}_{\IS} \approx  0.84$ for the MIXL model so that IS$^2$ is only 16\% less
efficient than IS. The corresponding figure is 0.37 for the GMNL model, so here IS$^2$ is 63\% less efficient than IS as we used a very simple importance
sampler.
\begin{table}[!tbp]
\begin{center}
\caption{The MIXL model: The table shows the relative variances for posterior mean estimation using IS$^2$ based on $M=50,000$ draws for $\theta$ and different $\sigma^2$.}\label{tab:mixlexample1}
\vspace{1pt}
\begin{threeparttable}
\begin{tabular}{lcccc}
\hline\hline
\multicolumn{1}{c}{} & $\sigma^2\approx 0.05$ & $\sigma^2\approx 0.17$ & $\sigma^2\approx 0.5$ & $\sigma^2\approx 0.75$ \\
\cline{2-5}
$\beta_{01}$ & 0.642 & 1.000 & 2.142 & 2.533 \\
$\beta_1$ & 0.665 & 1.000 & 2.098 & 2.301 \\
$\beta_2$ & 0.695 & 1.000 & 1.985 & 2.408 \\
$\beta_3$ & 0.579 & 1.000 & 1.705 & 1.982 \\
$\beta_4$ & 0.649 & 1.000 & 1.730 & 2.174 \\
$\beta_5$ & 0.630 & 1.000 & 1.605 & 3.060 \\
$\sigma_{01}$ & 0.590 & 1.000 & 2.232 & 2.668 \\
$\sigma_1$ & 0.597 & 1.000 & 1.774 & 2.103 \\
$\sigma_2$ & 0.635 & 1.000 & 1.898 & 2.440 \\
$\sigma_3$ & 0.619 & 1.000 & 1.802 & 3.006 \\
$\sigma_4$ & 0.926 & 1.000 & 1.420 & 2.757 \\
Average MC & 0.657 & 1.000 & 1.854 & 2.494 \\
Theoretical MC & 0.861 & 1.000 & 1.350 & 1.733 \\
Time (minutes) & 15.85 & 10.32 & 9.27 & 8.93 \\
Estimated TNV & 1.01 & 1.00 & 1.67 & 2.16 \\
Theoretical TNV & 1.25 & 1.00 & 1.23 & 1.54 \\
\hline\hline
\end{tabular}
\end{threeparttable}
\end{center}
\end{table}

\begin{table}[!tbp]
\begin{center}
\caption{The GMNL model: The table shows the relative variances for posterior mean estimation using IS$^2$ based on $M=50,000$ draws for $\theta$ and different $\sigma^2$.}\label{tab:example1}
\vspace{1pt}
\begin{threeparttable}
\vspace{4pt}
\begin{tabular}{lcccc}
\hline\hline
\multicolumn{1}{c}{} & $\sigma^2\approx 0.5$ & $\sigma^2\approx 1$ & $\sigma^2\approx 1.5$ & $\sigma^2\approx 2$ \\
\cline{2-5}
$\beta_0$ & 0.618 & 1.000 & 1.200 & 2.045 \\
$\beta_1$ & 0.744 & 1.000 & 1.295 & 2.220 \\
$\beta_2$ & 0.766 & 1.000 & 1.440 & 1.736 \\
$\beta_3$ & 0.873 & 1.000 & 1.553 & 2.398 \\
$\beta_4$ & 0.714 & 1.000 & 1.368 & 2.019 \\
$\beta_5$ & 0.658 & 1.000 & 1.198 & 2.060 \\
$\sigma_0$ & 0.638 & 1.000 & 1.342 & 1.620 \\
$\sigma_1$ & 0.614 & 1.000 & 1.490 & 2.061 \\
$\sigma_2$ & 0.729 & 1.000 & 1.412 & 1.942 \\
$\sigma_3$ & 0.884 & 1.000 & 1.507 & 2.435 \\
$\sigma_4$ & 0.824 & 1.000 & 1.425 & 2.229 \\
$\delta$ & 0.700 & 1.000 & 1.383 & 2.005 \\
$\gamma$ & 0.686 & 1.000 & 1.324 & 1.776 \\
Average MC & 0.727 & 1.000 & 1.380 & 2.042 \\
Theoretical MC & 0.607 & 1.000 & 1.649 & 2.718 \\
Time (minutes) & 535 & 262 & 178 & 138 \\
Estimated TNV & 1.486 & 1.000 & 0.939 & 1.081 \\
Theoretical TNV & 1.213 & 1.000 & 1.099 & 1.359 \\
\hline\hline
\end{tabular}
\end{threeparttable}
\end{center}
\end{table}

\begin{table}[h]
\begin{center}
\caption{The table presents estimates of selected posterior statistics based on $M=50,000$ importance samples for the parameters. The MC standard errors are in brackets.
$\widehat{\ESS}_{\IS}$ means the estimate of the equivalent sample size for the IS method calculated and interpreted as in section~\ref{S: ESS} of the supplementary material. The corresponding
$\widehat{\ESS}_{\IS^2} \approx  \widehat{\ESS}_{\IS}/ \exp( \sigma_{\rm opt}^2)$, where $\sigma_{\rm opt}^2$ is  the optimal variance of the log likelihood.  }
\label{tab:posterior}
\begin{threeparttable}
\begin{tabular}{lcccclcccc}
\hline\hline
 & \multicolumn{ 4}{c}{MIXL} &  & \multicolumn{ 4}{c}{GMNL} \\
 & Mean & SD & \multicolumn{ 2}{c}{90\% CI} &  & Mean & SD & \multicolumn{ 2}{c}{90\% CI} \\
\cline{2-5}\cline{7-10}
$\beta_{01}$ & $\underset{[0.002]}{-1.032}$ & $\underset{[0.002]}{0.430}$ & $\underset{[0.006]}{-1.731}$ & $\underset{[0.006]}{-0.324}$ &  & $\underset{[0.006]}{-1.183}$ & $\underset{[0.004]}{0.458}$ & $\underset{[0.010]}{-1.927}$ & $\underset{[0.015]}{-0.412}$ \\[10pt]
$\beta_1$ & $\underset{[0.001]}{0.831}$ & $\underset{[0.001]}{0.234}$ & $\underset{[0.003]}{0.451}$ & $\underset{[0.003]}{1.223}$ &  & $\underset{[0.005]}{1.092}$ & $\underset{[0.007]}{0.348}$ & $\underset{[0.006]}{0.584}$ & $\underset{[0.014]}{1.702}$ \\[10pt]
$\beta_2$ & $\underset{[0.001]}{-1.405}$ & $\underset{[0.002]}{0.352}$ & $\underset{[0.005]}{-1.986}$ & $\underset{[0.005]}{-0.827}$ &  & $\underset{[0.006]}{-1.763}$ & $\underset{[0.005]}{0.480}$ & $\underset{[0.013]}{-2.596}$ & $\underset{[0.007]}{-1.039}$ \\[10pt]
$\beta_3$ & $\underset{[0.002]}{3.249}$ & $\underset{[0.002]}{0.400}$ & $\underset{[0.005]}{2.618}$ & $\underset{[0.007]}{3.931}$ &  & $\underset{[0.014]}{4.091}$ & $\underset{[0.021]}{0.812}$ & $\underset{[0.013]}{3.023}$ & $\underset{[0.043]}{5.562}$ \\[10pt]
$\beta_4$ & $\underset{[0.001]}{1.323}$ & $\underset{[0.001]}{0.224}$ & $\underset{[0.003]}{0.963}$ & $\underset{[0.003]}{1.696}$ &  & $\underset{[0.006]}{1.658}$ & $\underset{[0.008]}{0.380}$ & $\underset{[0.012]}{1.106}$ & $\underset{[0.023]}{2.358}$ \\[10pt]
$\beta_5$ & $\underset{[0.000]}{-0.219}$ & $\underset{[0.000]}{0.069}$ & $\underset{[0.001]}{-0.334}$ & $\underset{[0.001]}{-0.106}$ &  & $\underset{[0.001]}{-0.271}$ & $\underset{[0.001]}{0.093}$ & $\underset{[0.003]}{-0.430}$ & $\underset{[0.002]}{-0.128}$ \\[10pt]
$\sigma_{01}$ & $\underset{[0.002]}{3.015}$ & $\underset{[0.002]}{0.419}$ & $\underset{[0.004]}{2.386}$ & $\underset{[0.007]}{3.742}$ &  & $\underset{[0.006]}{3.121}$ & $\underset{[0.005]}{0.464}$ & $\underset{[0.009]}{2.418}$ & $\underset{[0.015]}{3.927}$ \\[10pt]
$\sigma_1$ & $\underset{[0.001]}{1.453}$ & $\underset{[0.001]}{0.228}$ & $\underset{[0.002]}{1.103}$ & $\underset{[0.003]}{1.845}$ &  & $\underset{[0.005]}{1.693}$ & $\underset{[0.006]}{0.360}$ & $\underset{[0.005]}{1.195}$ & $\underset{[0.013]}{2.320}$ \\[10pt]
$\sigma_2$ & $\underset{[0.002]}{2.564}$ & $\underset{[0.002]}{0.377}$ & $\underset{[0.003]}{1.996}$ & $\underset{[0.006]}{3.231}$ &  & $\underset{[0.009]}{2.856}$ & $\underset{[0.014]}{0.587}$ & $\underset{[0.008]}{2.075}$ & $\underset{[0.034]}{3.938}$ \\[10pt]
$\sigma_3$ & $\underset{[0.001]}{2.656}$ & $\underset{[0.002]}{0.352}$ & $\underset{[0.003]}{2.124}$ & $\underset{[0.007]}{3.277}$ &  & $\underset{[0.009]}{3.005}$ & $\underset{[0.012]}{0.592}$ & $\underset{[0.007]}{2.195}$ & $\underset{[0.024]}{4.057}$ \\[10pt]
$\sigma_4$ & $\underset{[0.001]}{1.136}$ & $\underset{[0.001]}{0.233}$ & $\underset{[0.004]}{0.768}$ & $\underset{[0.003]}{1.531}$ &  & $\underset{[0.005]}{1.247}$ & $\underset{[0.005]}{0.337}$ & $\underset{[0.013]}{0.747}$ & $\underset{[0.010]}{1.836}$ \\[10pt]
$\delta$ & \multicolumn{1}{l}{} & \multicolumn{1}{l}{} & \multicolumn{1}{l}{} & \multicolumn{1}{l}{} &  & $\underset{[0.004]}{0.621}$ & $\underset{[0.004]}{0.183}$ & $\underset{[0.020]}{0.318}$ & $\underset{[0.007]}{0.928}$ \\[10pt]
$\gamma$ & \multicolumn{1}{l}{} & \multicolumn{1}{l}{} & \multicolumn{1}{l}{} & \multicolumn{1}{l}{} &  & $\underset{[0.004]}{0.170}$ & $\underset{[0.008]}{0.144}$ & $\underset{[0.001]}{0.014}$ & $\underset{[0.018]}{0.452}$ \\[10pt]
$\log\hspace{1mm}p(y)$& $\underset{[0.003]}{-981.24}$ &&&&&$\underset{[0.012]}{-978.27}$\\[10pt]
$\widehat{\ESS}_{\IS^2}/\widehat{\ESS}_{\IS}$ &  0.84  &&&&& 0.37\\
\hline\hline
\end{tabular}
\end{threeparttable}
\end{center}
\end{table}
\bibliographystyle{apalike}
\bibliography{references,marcelbib}

\begin{thebibliography}{}

\bibitem[Albert and Chib, 1993]{albert:chib:1993}
Albert, J.~H. and Chib, S. (1993).
\newblock Bayesian analysis of binary and polychotomous response data.
\newblock {\em Journal of the American Statistical Association},
  88(422):669--679.

\bibitem[Andrieu et~al., 2010]{Andrieu:2010}
Andrieu, C., Doucet, A., and Holenstein, R. (2010).
\newblock Particle {Markov chain Monte Carlo} methods.
\newblock {\em Journal of the Royal Statistical Society, Series B}, 72:1--33.

\bibitem[Andrieu and Roberts, 2009]{Andrieu:2009}
Andrieu, C. and Roberts, G. (2009).
\newblock The pseudo-marginal approach for efficient {Monte Carlo}
  computations.
\newblock {\em The Annals of Statistics}, 37:697--725.

\bibitem[Barndorff-Nielsen and Shephard, 2002]{bNnS2002}
Barndorff-Nielsen, O. and Shephard, N. (2002).
\newblock Econometric analysis of realized volatility and its use in estimating
  stochastic volatility models.
\newblock {\em Journal of the Royal Statistical Society B}, 64:253--280.

\bibitem[Bazaraa et~al., 2006]{Bazaraa:2006}
Bazaraa, M.~S., Sherali, H.~D., and Shetty, C.~M. (2006).
\newblock {\em Nonlinear Programming}.
\newblock Wiley, New Jersey, 3rd edition.

\bibitem[Beaumont, 2003]{Beaumont:2003}
Beaumont, M.~A. (2003).
\newblock Estimation of population growth or decline in genetically monitored
  populations.
\newblock {\em Genetics}, 164:1139--1160.

\bibitem[Chib and Jeliazkov, 2001]{Chib:2001}
Chib, S. and Jeliazkov, I. (2001).
\newblock {Marginal likelihood from the Metropolis Hastings output}.
\newblock {\em Journal of the American Statistical Association},
  96(453):270--281.

\bibitem[Chopin et~al., 2013]{Chopin:2012}
Chopin, N., Jacob, P.~E., and Papaspiliopoulos, O. (2013).
\newblock {SMC{\textasciicircum}2:} an efficient algorithm for sequential
  analysis of state-space models.
\newblock {\em Journal of the Royal Statistical Society B}, 75(3):397--426.

\bibitem[Doucet et~al., 2015]{Doucet:2014}
Doucet, A., Pitt, M.~K., Deligiannidis, G., and Kohn, R. (2015).
\newblock {Efficient implementation of Markov chain Monte Carlo when using an
  unbiased likelihood estimator}.
\newblock {\em Biometrika}, 102(2):295--313.

\bibitem[Duan and Fulop, 2015]{duan:fulop:2015}
Duan, J.-C. and Fulop, A. (2015).
\newblock Density-tempered marginalized sequential monte carlo samplers.
\newblock {\em Journal of Business \& Economic Statistics}, 33(2):192--202.

\bibitem[Fearnhead et~al., 2008]{Fearnhead2008}
Fearnhead, P., Papaspiliopoulos, O., and Roberts, G.~O. (2008).
\newblock {Particle filters for partially observed diffusions}.
\newblock {\em Journal of the Royal Statistical Society: Series B (Statistical
  Methodology)}, 70(4):755--777.

\bibitem[Fearnhead et~al., 2010]{Fearnhead2010}
Fearnhead, P., Papaspiliopoulos, O., Roberts, G.~O., and Stuart, A. (2010).
\newblock {Random-weight particle filtering of continuous time processes}.
\newblock {\em Journal of the Royal Statistical Society: Series B (Statistical
  Methodology)}, 72(4):497--512.

\bibitem[Fiebig et~al., 2010]{fklw2010}
Fiebig, D.~G., Keane, M.~P., Louviere, J., and Wasi, N. (2010).
\newblock {The Generalized Multinomial Logit Model: Accounting for Scale and
  Coefficient Heterogeneity}.
\newblock {\em Marketing Science}, 29(3):393--421.

\bibitem[Gelman, 2006]{Gelman2006}
Gelman, A. (2006).
\newblock {Prior distributions for variance parameters in hierarchical models}.
\newblock {\em Bayesian Analysis}, 1(3):515--534.

\bibitem[Geweke, 1989]{Geweke:1989}
Geweke, J. (1989).
\newblock Bayesian inference in econometric models using {Monte Carlo}
  integration.
\newblock {\em Econometrica}, 57(6):1317--1339.

\bibitem[Gourieroux and Monfort, 1995]{Gourieroux:1995}
Gourieroux, C. and Monfort, A. (1995).
\newblock {\em Statistics and Econometric Models}, volume~2.
\newblock Cambridge University Press, Melbourne.

\bibitem[Hesterberg, 1995]{Hesterberg1995}
Hesterberg, T. (1995).
\newblock Weighted average importance sampling and defensive mixture
  distributions.
\newblock {\em Technometrics}, 37(2):185--194.

\bibitem[Hoogerheide et~al., 2012]{hod2012}
Hoogerheide, L., Opschoor, A., and van Dijk, H.~K. (2012).
\newblock {A class of adaptive importance sampling weighted EM algorithms for
  efficient and robust posterior and predictive simulation}.
\newblock {\em Journal of Econometrics}, 171(2):101--�120.

\bibitem[Johndrow et~al., 2016]{johndrow2016inefficiency}
Johndrow, J.~E., Smith, A., Pillai, N., and Dunson, D.~B. (2016).
\newblock Inefficiency of data augmentation for large sample imbalanced data.
\newblock {\em arXiv preprint arXiv:1605.05798}.

\bibitem[Kass and Raftery, 1995]{Kass:1995}
Kass, R.~E. and Raftery, A.~E. (1995).
\newblock Bayes factors.
\newblock {\em Journal of the American Statistical Association},
  90(430):773--795.

\bibitem[Liu, 2001]{Liu:2001}
Liu, J.~S. (2001).
\newblock {\em Monte Carlo Strategies in Scientific Computing}.
\newblock Springer-Verlag, New York.

\bibitem[Neal, 2001]{Neal:2001}
Neal, R. (2001).
\newblock Annealed importance sampling.
\newblock {\em Statistics and Computing}, 11:125--139.

\bibitem[Perrakis et~al., 2014]{Perrakis:2014}
Perrakis, K., Ntzoufras, I., and Tsionas, E.~G. (2014).
\newblock On the use of marginal posteriors in marginal likelihood estimation
  via importance sampling.
\newblock {\em Computational Statistics \& Data Analysis}, 77(0):54 -- 69.

\bibitem[Pitt et~al., 2012]{Pitt:2012}
Pitt, M.~K., Silva, R.~S., Giordani, P., and Kohn, R. (2012).
\newblock On some properties of {Markov chain Monte Carlo} simulation methods
  based on the particle filter.
\newblock {\em Journal of Econometrics}, 171(2):134--151.

\bibitem[Polson et~al., 2013]{polson:scott:windle:2013}
Polson, N.~G., Scott, J.~G., and Windle, J. (2013).
\newblock Bayesian inference for logistic models using pólya–gamma latent
  variables.
\newblock {\em Journal of the American Statistical Association},
  108(504):1339--1349.

\bibitem[Quiroz et~al., 2015]{quiroz:2015}
Quiroz, M., Villani, M., and Kohn, R. (2015).
\newblock Scalable {MCMC} for large data problems using data subsampling and
  the difference estimator.
\newblock \url{http://arxiv.org/abs/1507.02971}.

\bibitem[Richard and Zhang, 2007]{ZR2007}
Richard, J.-F. and Zhang, W. (2007).
\newblock Efficient high-dimensional importance sampling.
\newblock {\em Journal of Econometrics}, 141(2):1385--1411.

\bibitem[Ripley, 1987]{Ripley87}
Ripley, B. (1987).
\newblock {\em Stochastic Simulation}.
\newblock Wiley, New York.

\bibitem[Scharth and Kohn, 2016]{PEIS}
Scharth, M. and Kohn, R. (2016).
\newblock Particle efficient importance sampling.
\newblock {\em Journal of Econometrics}, 190(1):133 -- 147.

\bibitem[Stirzaker and Grimmett, 2001]{stirzaker2001probability}
Stirzaker, G. G.~D. and Grimmett, D. (2001).
\newblock {\em Probability and random processes}, volume~19.
\newblock Oxford University Press, Oxford, UK.

\bibitem[Tran et~al., 2014]{tran2014_aisel}
Tran, M.-N., Pitt, M.~K., and Kohn, R. (2014).
\newblock Annealed important sampling for models with latent variables.
\newblock \url{http://arxiv.org/abs/1402.6035}.

\bibitem[Vehtari and Gelman, 2015]{Vehtari2015}
Vehtari, A. and Gelman, A. (2015).
\newblock {Pareto Smoothed Importance Sampling}.
\newblock \url{http://arxiv.org/abs/1507.02646}.

\end{thebibliography}
\end{document}